\documentclass[11pt,reqno]{amsart}
\usepackage{epsfig}
\usepackage{verbatim}
\usepackage{color}
\usepackage{graphicx,amssymb,amsmath,amsthm}
\usepackage[colorlinks, linkcolor=red, anchorcolor=blue, citecolor=green]{hyperref}
\usepackage{mathrsfs}
\usepackage{amssymb}
\usepackage{enumerate}
\usepackage{pifont}
\usepackage{chngcntr}
\usepackage{enumitem}

\newcommand{\BH}{{\mathbb H}}

\newcommand{\R}{{\mathbb R}}

\newcommand{\C}{{\mathbb C}}

\renewcommand{\eqref}[1]{(\ref{#1})}

\newcommand{\abs}[1]{\lvert#1\rvert}

\newcommand{\dist}{{\rm dist}}

\newcommand{\vx}{{\boldsymbol x}}
\newcommand{\vy}{{\boldsymbol y}}

\newcommand{\va}{{\boldsymbol a}}

\newcommand{\BA}{{\boldsymbol{A}}}

\renewcommand{\top}{T}
\newtheorem{prop}{Proposition}[section]
\newtheorem{lem}[prop]{Lemma}

\newcommand{\innerp}[1]{\langle {#1} \rangle}

\newtheorem{coro}[prop]{Corollary}
\newtheorem{theorem}[prop]{Theorem}

\newtheorem{remark}[prop]{Remark}

\newtheorem{conjecture}[prop]{Conjecture}

\newcommand{\argmin}[1]{\mathop{\rm argmin}\limits_{#1}}

\newcommand\zfrac[2]{%
\genfrac{}{}{0.5pt}{0}%
{#1}{#2}}

\numberwithin{equation}{section}
\date{}

\newcommand{\Addresses}{{
\bigskip
\footnotesize

Yu Xia, \textsc{Department of Mathematics, Hangzhou Normal University, Hangzhou 311121, China}\par\nopagebreak
\textit{E-mail address}, Yu Xia: \texttt{yxia@hznu.edu.cn}

\medskip

Zhiqiang Xu, \textsc{LSEC, Inst.~Comp.~Math., Academy of
Mathematics and System Science,  Chinese Academy of Sciences, Beijing, 100091, China
\newline
School of Mathematical Sciences, University of Chinese Academy of Sciences, Beijing 100049, China}\par\nopagebreak
\textit{E-mail address}, Zhiqiang Xu: \texttt{xuzq@lsec.cc.ac.cn}

\medskip

Zili Xu, \textsc{Department of Mathematics, The Hong Kong University of Science and Technology, Clear Water Bay, Kowloon, Hong Kong SAR, China}\par\nopagebreak
\textit{E-mail address}, Zili Xu: \texttt{xuzili@ust.hk}
}}

\begin{document}
\bibliographystyle{plain}
\title[ Stability in Phase Retrieval]{ Stability in Phase Retrieval: Characterizing Condition Numbers and the Optimal Vector Set}

\author{Yu Xia, Zhiqiang Xu, and Zili Xu}

\thanks{Zhiqiang Xu was supported in part
 by the National Science Fund for Distinguished Young Scholars under Grant
 12025108 and in part by the National Natural Science Foundation of China
 under Grant 12021001 and Grant 12288201. Yu Xia was supported by NSFC grant (12271133, U21A20426, 11901143) and the key project of Zhejiang Provincial Natural Science Foundation grant (LZ23A010002)}

%
%

\maketitle

\begin{abstract}
 In this paper, we primarily focus on analyzing the stability property of phase retrieval by examining the bi-Lipschitz property of the map $\Phi_{\boldsymbol{A}}(\boldsymbol{x})=|\boldsymbol{A}\boldsymbol{x}|\in \mathbb{R}_+^m$,
where $\boldsymbol{x}\in \mathbb{H}^d$ and $\boldsymbol{A}\in \mathbb{H}^{m\times d}$ is the measurement matrix for $\mathbb{H}\in\{\mathbb{R},\mathbb{C}\}$. We define the condition number $\beta_{\boldsymbol{A}}=\frac{U_{\boldsymbol{A}}}{L_{\boldsymbol{A}}}$, where $L_{\boldsymbol{A}}$ and $U_{\boldsymbol{A}}$ represent the optimal lower and upper Lipschitz constants, respectively. We establish the  universal lower bound on $\beta_{\boldsymbol{A}}$ by demonstrating that for any ${\boldsymbol{A}}\in\mathbb{H}^{m\times d}$, 
\begin{equation*}
\beta_{\boldsymbol{A}}\geq \beta_0^{\mathbb{H}}=\begin{cases}
\sqrt{\frac{\pi}{\pi-2}}\,\,\approx\,\, 1.659  & \text{if $\mathbb{H}=\mathbb{R}$,}\\
\sqrt{\frac{4}{4-\pi}}\,\,\approx\,\, 2.159  & \text{if $\mathbb{H}=\mathbb{C}$.}
\end{cases}
\end{equation*} 
We prove that the condition number of a standard Gaussian matrix in $\mathbb{H}^{m\times d}$ asymptotically matches the lower bound $\beta_0^{\mathbb{H}}$ for both real and complex cases.
This result indicates that the constant lower bound $\beta_0^{\mathbb{H}}$ is asymptotically tight, holding true for both the real and complex scenarios.
As an application of this result, we utilize it to investigate the performance of quadratic models for phase retrieval.
Lastly, we establish that for any odd integer $m\geq 3$, the harmonic frame $\boldsymbol{E}_m\in \mathbb{R}^{m\times 2}$ possesses the minimum condition number among all $\boldsymbol{A}\in \mathbb{R}^{m\times 2}$.

To the best of our knowledge, our findings provide the first universal lower bound for the condition number in phase retrieval. Additionally, we have identified the first optimal vector set in $\mathbb{R}^2$ for phase retrieval. We are confident that these findings carry substantial implications for enhancing our understanding of phase retrieval.
\end{abstract}

\vspace{0.4cm}

\section{introduction}
\setcounter{equation}{0}

\subsection{Phase retrieval}
Assume that $\boldsymbol{A}=(\va_1,\ldots,\va_m)^*\in \BH^{m\times d}$, where $\va_j\in \BH^d$ are known vectors and $\BH\in \{\R, \C\}$.
The aim of phase retrieval is to recover $\vx\in \BH^d$ from the phaseless measurements $\abs{\innerp{\va_j,\vx}}, j=1,\ldots,m$.  For convenience, we define the nonlinear map $\Phi_{\boldsymbol{A}}:\BH^d\to \mathbb{R}_+^m$ as
\begin{equation*}
\Phi_{\boldsymbol{A}}(\boldsymbol{x})=\abs{\boldsymbol{A}\vx}:=(\abs{\innerp{\va_1,\vx}}, \abs{\innerp{\va_2,\vx}},\ldots,\abs{\innerp{\va_m,\vx}})^T\in \mathbb{R}_+^m.
\end{equation*}
We say $\boldsymbol{A}$ has {\em phase retrieval property} if $\abs{\boldsymbol{A} \vx}=\abs{\boldsymbol{A}\vy}$ implies $\vx=c\cdot \vy$ for some $c\in \BH$ with $\abs{c}=1$.
The existing literature has outlined certain conditions on $\boldsymbol{A}$ that ensure the phase retrieval property, as demonstrated in \cite{BCMN, 4d4, WangXu}. Specifically, it has been established that $m\geq 2d-1$ (or $m\geq 4d-4$) generic measurements are adequate for the precise recovery of $\vx\in \BH^d$, up to a unimodular constant, where $\BH=\R$ (or $\BH=\C$).

\subsection{Stability property of phase retrieval}

The stability of signal reconstruction is of utmost importance in the domain of signal recovery from phaseless measurements.  It not only bolsters the robustness of the reconstruction process but also preserves the essential characteristic of injectivity. In order to achieve reliable signal recovery, the literature on non-convex approaches to phase retrieval typically imposes a stability condition. Notably, research papers such as \cite{BW, BCMN, EM14, siamre, stablefunc} highlight the significance of this stability condition.

One way to quantify the robustness of the phase retrieval process for a given measure matrix ${\boldsymbol{A}}\in \BH^{m\times d}$ is in terms of the Lipschitz bound of the map $\Phi_{\boldsymbol{A}}$.
For any $\vx,\vy\in \BH^d$, we define the distance between $\vx$ and $\vy$ as  
\begin{equation}\label{dist-def2024}
\dist_\BH(\vx,\vy):=\min \{\|\vx-c\cdot \vy\|_2: c\in \BH, \abs{c}=1\}.
\end{equation}
Assume that the measure matrix $\boldsymbol{A}=(\va_1,\ldots,\va_m)^*\in \BH^{m\times d}$ has phase retrieval property.  
It has been demonstrated that the map $\Phi_{\boldsymbol{A}}$ is bi-Lipschitz \cite{BCMN,BW,CCD16,AG, BZ}, that is, there exist two positive constants $0<L\leq U<\infty$ such that for any $\vx,\vy\in\mathbb{H}^d$,
\begin{equation}\label{Lipschitz}
L \cdot \dist_\BH(\vx,\vy)\,\,\leq\,\, \|\abs{\boldsymbol{A}\vx}-\abs{\boldsymbol{A}\vy}\|_2\leq U\cdot\dist_\BH(\vx,\vy).
\end{equation}
A comprehensive overview of this topic is provided in \cite{siamre}. 
We denote the greatest possible $L$ and the smallest possible $U$ as $L_{\boldsymbol{A}}^\BH$  and $U_{\boldsymbol{A}}^\BH$, respectively. In other words, we set
\begin{equation*}
L_{\boldsymbol{A}}^\BH:= \inf_{\substack{\vx,\vy\in \mathbb{H}^d\\ \dist_\BH{(\vx,\vy)}\neq 0}}\frac{\|\abs{\boldsymbol{A}\vx}-\abs{\boldsymbol{A}\vy}\|_2}{\dist_\BH(\vx,\vy)}\quad\text{and}\quad
U_{\boldsymbol{A}}^\BH:= \sup_{\substack{\vx,\vy\in \mathbb{H}^d\\ \dist_\BH{(\vx,\vy)}\neq 0}}\frac{\|\abs{\boldsymbol{A}\vx}-\abs{\boldsymbol{A}\vy}\|_2}{\dist_\BH(\vx,\vy)}.
\end{equation*}
Numerous studies have been conducted to estimate or determine the optimal Lipschitz bounds $L_{\boldsymbol{A}}^\BH$ and $U_{\boldsymbol{A}}^\BH$ \cite{BCMN,BW,CCD16,AG,siamre}.

In this paper, our primary focus is on the {\em condition number} $\beta_{\boldsymbol{A}}^\BH$, which is defined as
\begin{equation*}
\beta_{\boldsymbol{A}}^\BH\,\,:=\,\, \frac{U_{\boldsymbol{A}}^\BH}{L_{\boldsymbol{A}}^\BH}.
\end{equation*}
The quantity $\beta_{\boldsymbol{A}}^{\BH}$ is referred to as the {\em distortion} in \cite{CIM23}.
The condition number serves as a measure of the stability of the measure matrix $\boldsymbol{A}$ for phase retrieval.
 If $\boldsymbol{A}$ lacks the phase retrieval property, i.e., ${L_{\boldsymbol{A}}^\BH}=0$, then we set $\beta_{\boldsymbol{A}}^\BH=+\infty$.
 On the other hand, if $\boldsymbol{A}$ possesses the phase retrieval property, $\beta_{\boldsymbol{A}}^\BH$ becomes a finite positive number. A smaller condition number $\beta_{\boldsymbol{A}}^\BH$ indicates that $\Phi_{\boldsymbol{A}}$ behaves more like a near-isometry. The objective of this paper is to analyze the stability of a given measure matrix $\boldsymbol{A}\in\mathbb{H}^{m\times d}$ by examining its condition number. Particularly, we are interested in the following questions:
\vspace{0.3cm} 
\begin{enumerate}[align=parleft, labelwidth=2cm]
\item[Question I] Does there exists a universal lower bound, denoted as $\beta_0^\BH>1$, so that $\beta_{\boldsymbol{A}}^\BH\geq \beta_0^\BH$ for all matrices $\boldsymbol{A}\in \BH^{m\times d}$?
\item [Question II]
For a given pair of integers $m$ and $d$, what is the optimal measurement matrix that has the minimal condition number?
In other words, can we identify a matrix $\boldsymbol{E}\in \BH^{m\times d}$ such that $\boldsymbol{E}\in {\rm argmin}_{\boldsymbol{A}\in \BH^{m\times d}} \beta_{\boldsymbol{A}}^\BH$?
\end{enumerate}
\vspace{0.3cm} 

To streamline the notation, we often omit the superscript of $\beta_{\boldsymbol{A}}^\BH$ and determine whether $\beta_{\boldsymbol{A}}$ is defined in a real or complex space solely based on the nature of the matrix $\boldsymbol{A}$, whether it is real or complex. Similarly, we can drop the superscript of $L_{\boldsymbol{A}}^\BH$, $U_{\boldsymbol{A}}^\BH$, $\beta_0^\BH$ and omit the subscript of $\dist_\BH(\vx,\vy)$.

\subsection{Our contribution}
The aim of this paper is trying to answer Question I and Question II.
We first give an affirmative answer to Question I by showing that there exists a constant lower bound on the condition number $\beta_{\boldsymbol{A}}$ for all $\boldsymbol{A}\in \mathbb{H}^{m\times d}$. Specifically, in Theorem \ref{lowerbound-ndim} of Section 2, we prove that for any $\boldsymbol{A}\in \mathbb{H}^{m\times d}$, 
\begin{equation}
\label{eq:mdlower1}
\beta_{\boldsymbol{A}}\geq \beta_0^{\mathbb{H}}:=\begin{cases}
\sqrt{\frac{\pi}{\pi-2}}\,\,\approx\,\, 1.659  & \text{if $\mathbb{H}=\mathbb{R}$,}\\
\sqrt{\frac{4}{4-\pi}}\,\,\approx\,\, 2.159  & \text{if $\mathbb{H}=\mathbb{C}$.}\\

\end{cases}
\end{equation} 	
To the best of our knowledge,  this result provides the first known constant lower bound on the condition number $\Phi_{\boldsymbol{A}}$ for both real and complex cases. 
Additionally, we show that this constant lower bound is asymptotically tight for both real and complex cases when $m\to\infty$. Specifically, in Theorem \ref{cos_lemma} of Section \ref{section:dim2} we calculate the condition number $\beta_{\boldsymbol{E}_m}$ of the harmonic frame $\boldsymbol{E}_m$ in $\R^2$:
\begin{equation}
\label{betaA-R2-2024xu}
\beta_{\boldsymbol{E}_m}=\begin{cases}
\frac{1}{\sqrt{1-\frac{2}{m\cdot \sin\frac{\pi}{m}}}}  & \text{if $m$ is even,}\\	
\frac{1}{\sqrt{1-\frac{1}{m\cdot \sin\frac{\pi}{2m}}}} & \text{if $m$ is odd.}
\end{cases}	
\end{equation}
Here, the harmonic frame $\boldsymbol{E}_m\in\R^{m\times 2}$ is defined as
\begin{equation*}
\boldsymbol{E}_m:=
\begingroup 
\arraycolsep=5pt\def\arraystretch{1.4}
\begin{pmatrix}
1 &  \cos(\frac{1}{m}\pi) & \cdots & \cos(\frac{m-1}{m}\pi)  \\
0 & \sin(\frac{1}{m}\pi) &\cdots &  \sin(\frac{m-1}{m}\pi) \\
\end{pmatrix}
\endgroup
^T.
\end{equation*}
From \eqref{betaA-R2-2024xu} we see that $\beta_{\boldsymbol{E}_m}$ asymptotically matches the lower bound $\beta_0^{\mathbb{R}}=\sqrt{\frac{\pi}{\pi-2}}$ as $m\to\infty$, confirming the tightness of $\beta_0^{\mathbb{R}}$ in the real case. 
Furthermore, in Section \ref{section:Gauss}, we establish that if $\boldsymbol{A}\in \mathbb{H}^{m\times d}$ is a standard Gaussian random matrix, where $\mathbb{H}\in\{\mathbb{R},\mathbb{C}\}$, then $\beta_{\boldsymbol{A}}$ approaches $\beta_0^{\mathbb{H}}$ asymptotically as $m\to\infty$. 
These results demonstrate that $\beta_0^{\mathbb{H}}$ is an asymptotically tight lower bound for both real and complex cases. 
 As a  application of this result, we employ it to examine the efficacy of quadratic models for phase retrieval  (see Corollary \ref{co:budeng}).

We next turn to Question II and mainly focus on the real case, i.e., $\mathbb{H}=\mathbb{R}$. In Theorem \ref{lowerbound-2dim}, we  improve the constant lower bound $\beta_0^{\mathbb{R}}=\sqrt{\frac{\pi}{\pi-2}}$ by showing that for any $\boldsymbol{A}\in\mathbb{R}^{m\times d}$ we have
\begin{equation}	\label{betterbound2024}
\beta_{\boldsymbol{A}}\,\,\geq\,\,  \frac{1}{\sqrt{1-\frac{1}{m\cdot \sin\frac{\pi}{2m}}}}.
\end{equation} 
Combining with \eqref{betaA-R2-2024xu}, we see that $\beta_{\boldsymbol{E}_m}$ matches the above lower bound for each odd integer $m\geq 3$. Therefore, $\boldsymbol{E}_m$ has the minimal condition number for each odd integer $m\geq 3$, i.e., $\boldsymbol{E}_m\in {\rm argmin}_{\boldsymbol{A}\in \mathbb{R}^{m\times 2}} \beta_{\boldsymbol{A}}$. 
This  addresses Question II for the real case with $d=2$. We believe that these findings provide insights into the general case of $\boldsymbol{A}\in\mathbb{H}^{m\times d}$.
\subsection{Related work}

Let $\boldsymbol{A}\in\mathbb{H}^{m\times d}$ be a measurement matrix that has phase retrieval property. Recall that we define the map $\Phi_{\boldsymbol{A}}: \BH^d\to \mathbb{R}_+^m$ as
\begin{equation*}
\Phi_{\boldsymbol{A}}(\boldsymbol{x})=\abs{\boldsymbol{A}\vx}=(\abs{\innerp{\va_1,\vx}}, \abs{\innerp{\va_2,\vx}},\ldots,\abs{\innerp{\va_m,\vx}})^T\in \mathbb{R}_+^m.
\end{equation*}
For convenience, we also define $\Phi^2_{\boldsymbol{A}}: \BH^d\to \mathbb{R}_+^m$ as
\begin{equation*}
\Phi_{\boldsymbol{A}}^2(\boldsymbol{x})=\abs{\boldsymbol{A}\vx}^2:=(\abs{\innerp{\va_1,\vx}}^2, \abs{\innerp{\va_2,\vx}}^2,\ldots,\abs{\innerp{\va_m,\vx}}^2)^T\in \mathbb{R}_+^m.
\end{equation*}
Most of the existing literature studied the stability of the phase retrieval process for ${\boldsymbol{A}}$ by analyzing the Lipschitz property of the map $\Phi_{\boldsymbol{A}}$ or $\Phi^2_{\boldsymbol{A}}$ with respect to different norms and metrics on the space $\BH^d$ \cite{EM14,BCMN,BW, BZ, CCD16,AG,Duchi,DDP20}. In the following we briefly introduce the existing results on the stability of phase retrieval.

\subsubsection{Lipschitz property of $\Phi_{\boldsymbol{A}}$}

Consider the real case $\mathbb{H}=\mathbb{R}$ first. Bandeiraa-Cahillb-Mixon-Nelson first study the bi-Lipschitz property \eqref{Lipschitz} of $\Phi_{\boldsymbol{A}}$ by estimating the optimal lower and upper Lipschitz constant \cite{BCMN}.  
Specifically, they established that for any $\boldsymbol{A}\in\mathbb{R}^{m\times d}$ we have 
\begin{equation}\label{relatedwork-upperbound}
U_{\boldsymbol{A}}=\|\boldsymbol{A}\|_2.	
\end{equation}
and
\begin{equation}\label{estimate-lowerbound}
\sigma_{\boldsymbol{A}}\leq L_{\boldsymbol{A}}\leq \sqrt{2}\cdot \sigma_{\boldsymbol{A}},
\end{equation}
where 
\begin{equation}\label{eq:sigma}
\sigma_{\boldsymbol{A}} \,\,:=\,\, \min_{I\subset [m]}{\rm max} \Big\{\sqrt{\lambda_{\rm min}(\boldsymbol{A}_I^*\boldsymbol{A}_I)},\sqrt{\lambda_{\rm min}(\boldsymbol{A}_{I^c}^*\boldsymbol{A}_{I^c})} \Big\}.
\end{equation}
Here, $[m]:=\{1,\ldots,m\}$,  $\boldsymbol{A}_I=(\va_j)_{j\in I}^*\in \mathbb{H}^{\# I\times d}$ denotes the row submatrix of $\boldsymbol{A}$, and $I^c:=[m]\setminus I$. Additionally, we use $\lambda_{\rm min}(\cdot)$ and $\lambda_{\rm max}(\cdot)$ to denote the minimal and maximal eigenvalues of a given Hermitian matrix, respectively. 
Their results immediately imply the following estimate on the condition number $\beta_{\boldsymbol{A}}$:
\begin{equation}\label{2024-neta-Mixon}
\frac{\|\boldsymbol{A}\|_2}{\sqrt{2}\cdot \sigma_{\boldsymbol{A}}} \leq \beta_{\boldsymbol{A}}=\frac{U_{\boldsymbol{A}}}{L_{\boldsymbol{A}}}=\frac{\|\boldsymbol{A}\|_2}{L_{\boldsymbol{A}}}\leq \frac{\|\boldsymbol{A}\|_2}{\sigma_{\boldsymbol{A}}}.
\end{equation}
Later, Balan-Wang \cite{BW} provided an exact value of the optimal lower Lipschitz constant:
\begin{equation*}
L_{\boldsymbol{A}}=\Delta_{\boldsymbol{A}} \,\,:=\,\, \min_{I\subset [m]} \sqrt{\lambda_{\rm min}(\boldsymbol{A}_I^*\boldsymbol{A}_I)+ \lambda_{\rm min}(\boldsymbol{A}_{I^c}^*\boldsymbol{A}_{I^c})}.
\end{equation*}
Consequently, we have
\begin{equation}\label{2024-beta-Wang}
\beta_{\boldsymbol{A}}=\frac{\|\boldsymbol{A}\|_2}{\Delta_{\boldsymbol{A}}}.
\end{equation}
It is worth noting  that $\sigma_{\boldsymbol{A}}\leq \Delta_{\boldsymbol{A}}\leq \sqrt{2}\cdot \sigma_{\boldsymbol{A}}$, making Balan-Wang's estimate in \eqref{2024-beta-Wang} an improvement over \eqref{2024-neta-Mixon}.

The complex case $\mathbb{H}=\mathbb{C}$ of the Lipschitz property of $\Phi_{\boldsymbol{A}}$ was later considered in \cite{CCD16,AG}. Both \cite{CCD16} and \cite{AG} considered the phase retrieval in a more general setting where the underlying space can be an infinite-dimensional Hilbert space. However, since our focus is on the complex space $\mathbb{C}^d$, we will only present their results for this case. The authors of \cite{CCD16} first showed that $\Phi_{\boldsymbol{A}}$ satisfies the Lipschitz property \eqref{Lipschitz} if $\boldsymbol{A}\in\mathbb{C}^{m\times d}$ has phase retrieval property. Specifically, they showed that $L_{\boldsymbol{A}}>0$ and $U_{\boldsymbol{A}}\leq \|\boldsymbol{A}\|_2$.
Later, the authors of \cite{AG} proved that $U_{\boldsymbol{A}}= \|\boldsymbol{A}\|_2$ and
\begin{equation*}
L_{\boldsymbol{A}}\leq 2\cdot \sqrt{\frac{\lambda_{\max}(\boldsymbol{A}^*\boldsymbol{A})}{\lambda_{\min}(\boldsymbol{A}^*\boldsymbol{A})}}\cdot \sigma_{\boldsymbol{A}}=\frac{2\cdot \|\boldsymbol{A}\|_2\cdot \sigma_{\boldsymbol{A}}}{\sqrt{\lambda_{\min}(\boldsymbol{A}^*\boldsymbol{A})}},
\end{equation*}
where $\sigma_{\boldsymbol{A}}$ is defined in \eqref{eq:sigma}. Therefore, in the complex case we have \cite[Corollary 3.10]{AG}
\begin{equation*}
\beta_{\boldsymbol{A}}\geq \frac{\sqrt{\lambda_{\min}(\boldsymbol{A}^*\boldsymbol{A})}}{2\cdot \sigma_{\boldsymbol{A}}}.
\end{equation*}

Recently, Alharbi et al. in \cite[Theorem 1.1]{AAFG24} demonstrated that for any matrix $\boldsymbol{A}\in\mathbb{H}^{m\times d}$, where $\mathbb{H}=\mathbb{R}$ or $\mathbb{C}$, the lower Lipschitz bound $L_{\boldsymbol{A}}$ can be expressed as follows:
\begin{equation*}
L_{\boldsymbol{A}}=\min_{\substack{\boldsymbol{x}\in \mathbb{H}^{d},\boldsymbol{y}\in \mathbb{H}^{d} \\ \|\boldsymbol{x}\|_2=1, \|\boldsymbol{y}\|_2\leq 1,\langle \boldsymbol{x},\boldsymbol{y}\rangle=0  }}
\frac{\||\boldsymbol{A}\boldsymbol{x}|-|\boldsymbol{A}\boldsymbol{y}|\|_2}
{\mathrm{dist}(\boldsymbol{x},\boldsymbol{y})}.
\end{equation*} 
Thus, to determine $L_{\boldsymbol{A}}$, it suffices to identify the minimum achievable ratio between 
$\||\boldsymbol{A}\boldsymbol{x}|-|\boldsymbol{A}\boldsymbol{y}|\|_2$ and $\mathrm{dist}(\boldsymbol{x},\boldsymbol{y})$ for orthogonal vectors $\boldsymbol{x},\boldsymbol{y}\in \mathbb{H}^d$.

\subsubsection{Lipschitz property of $\Phi_{\boldsymbol{A}}^2$}
The Lipschitz property of the map $\Phi_{\boldsymbol{A}}^2$ under the $\ell_1$-norm $\|\cdot \|_1$ was studied in \cite{EM14,Duchi,DDP20}.
Consider the real case first. We say that a measurement matrix $\boldsymbol{A}\in\mathbb{R}^{m\times d}$ is $\lambda\geq0$ stable if 
\begin{equation}\label{Eldar-condition}
\| |\boldsymbol{A}\vx|^2-|\boldsymbol{A}\vy|^2 \|_1\geq \lambda \cdot \|\vx-\vy\|_2\cdot \|\vx+\vy\|_2,\quad \text{ for all } \vx,\vy\in\mathbb{R}^d.	
\end{equation}
Eldar and Mendelson \cite{EM14} proved that the above condition holds with high probability if $\boldsymbol{A}$ is a sub-gaussian matrix.
In the complex case where $\mathbb{H}=\mathbb{C}$, Duchi and Ruan \cite{Duchi} extended the condition \eqref{Eldar-condition} to the following form: 
\begin{equation}
\| |\boldsymbol{A}\vx|^2-|\boldsymbol{A}\vy|^2 \|_1\geq \lambda\cdot  \inf_{\theta}\|\vx-e^{i\theta}\vy\|_2 \cdot \sup_{\theta}\|\vx-e^{i\theta}\vy\|_2,\quad \text{ for all } \vx,\vy\in\mathbb{C}^d.	
\end{equation}
They provided stability guarantees for both the real and complex cases by considering general classes of random matrices.

Previous research has also examined the bi-Lipschitz property of the map $\Phi_{\boldsymbol{A}}^2$ under the $\ell_2$-norm $\|\cdot \|_2$  \cite{BCMN,BW,BZ}. In the real case $\mathbb{H}=\mathbb{R}$, the authors of \cite{BCMN} first demonstrated that $\Phi_{\boldsymbol{A}}^2$ is not bi-Lipschitz under the $\ell_2$-norm $\|\cdot \|_2$ with respect to the distance $\mathrm{dist}(\vx,\vy)$ defined in \eqref{dist-def2024}. Subsequently, Balan-Wang \cite{BW} established that $\Phi_{\boldsymbol{A}}^2$ does satisfy the bi-Lipschitz property for the distance metric defined as:
\begin{equation*}
{\mathrm{d}}(\vx,\vy):=\|\vx\vx^T-\vy\vy^T\|_*\overset{(a)}=\|\vx-\vy\|_2\|\vx+\vy\|_2.
\end{equation*}
Here, $\|\boldsymbol{X}\|_*$ denotes the nuclear norm of $\boldsymbol{X}$, which is the sum of its singular values, and the equality denoted by $(a)$ is derived from \cite[Lemma 4.4]{BW}.
More specifically,  in \cite[Theorem 4.5]{BW}, Balan-Wang proved that there exist two positive constants $0<\widetilde{L}\leq \widetilde{U}\leq\infty$ such that for any $\vx,\vy\in\mathbb{R}^d$, 
\begin{equation*}
\widetilde{L}\cdot {\mathrm{d}}(\vx,\vy)\leq \||\boldsymbol{A}\boldsymbol{x}|^2 - |\boldsymbol{A}\boldsymbol{y}|^2 \|_2\leq \widetilde{U} \cdot {\mathrm{d}}(\vx,\vy).
\end{equation*}
Moreover, the constants $\widetilde{L}$ and $\widetilde{U}$ can be taken as 
\begin{equation*}
\widetilde{L}=\min\limits_{\|\vx\|_2=\|\vy\|_2=1}\bigg(\sum_{j=1}^m |\langle \vx,\boldsymbol{a}_j \rangle|^2|\langle \vy,\boldsymbol{a}_j \rangle|^2\bigg)^{1/2}\quad\text{and}\quad \widetilde{U}=\max\limits_{\|\vx\|_2=1}\bigg(\sum_{j=1}^m |\langle \vx,\boldsymbol{a}_j \rangle|^4 \bigg)^{1/2}	.
\end{equation*}
The exploration of the stability of the function $\Phi_{\boldsymbol{A}}^2$ with respect to alternative distance metrics has also been undertaken for the complex case $\mathbb{H}=\mathbb{C}$.
 For more detailed information, please refer to \cite{BZ}.
\subsection{Comparison to previous work} 

\subsubsection{The general case}
The primary focus of this paper is to investigate the bi-Lipschitz property of the map $\Phi_{\boldsymbol{A}}$ and to estimate the condition number $\beta_{\boldsymbol{A}}$.
We start with comparing our results with the previous estimates on the condition number for a general measure matrix $\boldsymbol{A}\in\mathbb{H}^{m\times d}$. In the real case where $\mathbb{H}=\mathbb{R}$, the previous best estimate on the condition number is in  \cite{BW}, that is, 
\begin{equation*}
\beta_{\boldsymbol{A}}=\frac{\|\boldsymbol{A}\|_2}{\Delta_{\boldsymbol{A}}}.	
\end{equation*}
However, it is nontrivial to estimate $\Delta_{\boldsymbol{A}}$ and hence  the condition number remains unknown. In this paper, we provide the  constant lower bound $\beta_0^{\mathbb{R}}=\sqrt{\frac{\pi}{\pi-2}}\approx 1.659$ on the condition number $\beta_{\boldsymbol{A}}$, which is asymptotically tight as $m\to \infty$. Moreover, we slightly improve upon this constant lower bound with the bound presented in \eqref{betterbound2024}, which is attained when $m\geq 3$ is odd and $\boldsymbol{A}$ is a harmonic frame in $\mathbb{R}^2$. 

In the complex case where $\mathbb{H}=\mathbb{C}$, the previous best estimate on the condition number is \cite{AG}
\begin{equation*}
\beta_{\boldsymbol{A}}\geq \frac{\sqrt{\lambda_{\min}(\boldsymbol{A}^*\boldsymbol{A})}}{2\cdot \sigma_{\boldsymbol{A}}}.
\end{equation*}
In this paper, we present the first constant lower bound $\beta_0^{\mathbb{C}}=\sqrt{\frac{4}{4-\pi}}\approx 2.159$, which is also asymptotically tight as $m\to \infty$.  

\subsubsection{The special case of $\boldsymbol{A}$ being a Gaussian random matrix}
For the case where $\boldsymbol{A}\in\mathbb{H}^{m\times d}$ is a standard Gaussian matrix, we show that $\beta_{\boldsymbol{A}}$ approaches $\beta_0^{\mathbb{H}}$ with high probability as $m\to \infty$. To the best of our knowledge, this is the first estimate on the condition number of a standard complex Gaussian matrix $\boldsymbol{A}\in\mathbb{C}^{m\times d}$. In the real case where $\mathbb{H}=\mathbb{R}$, our estimation improves upon the result obtained by Bandeiraa-Cahillb-Mixon-Nelson in \cite{BCMN}. 
Specifically, they showed that for a standard real Gaussian matrix $\boldsymbol{A}\in\mathbb{R}^{m\times d}$, the inequality 
\begin{equation*}
\sigma_{\boldsymbol{A}}\geq \frac{m-2d+2}{\sqrt{2}\cdot  e^{1+\frac{\epsilon}{R-2}}\cdot 2^{\frac{R}{R-2}}\cdot \sqrt{m}}
\end{equation*}
holds with probability at least $1-\exp(-\epsilon d)$, where $R:=\frac{m}{d}$ is assumed to be greater than $2$ \cite[Theorem 20]{BCMN}.
It is worth noting that $\|\boldsymbol{A}\|_2\leq (1+\epsilon)(\sqrt{d}+\sqrt{m})$ holds with probability $1-2\exp(-\frac{\epsilon}{2}(\sqrt{d}+\sqrt{m})^2)$ if $\boldsymbol{A}\in\mathbb{R}^{m\times d}$ is a standard Gaussian matrix \cite{DS01}. 
Therefore, using \eqref{2024-neta-Mixon}, we obtain the following upper bound on the condition number:
\begin{equation*}
\beta_{\boldsymbol{A}}\leq \frac{\|\boldsymbol{A}\|_2}{\sigma_{\boldsymbol{A}}}\leq 
(1+\epsilon)(\sqrt{d}+\sqrt{m}) \cdot \frac{ e^{1+\frac{\epsilon}{R-2}}\cdot 2^{\frac{3R-2}{2(R-2)}}\cdot \sqrt{m}}{m-2d+2}=(1+\epsilon)\cdot  e^{1+\frac{\epsilon}{R-2}}\cdot 2^{\frac{3R-2}{2(R-2)}} \cdot \frac{1+\frac{1}{\sqrt{R}}}{1-\frac{2}{R}+\frac{2}{m}},
\end{equation*}
which holds with probability at least $1-\exp(-\epsilon d)-2\exp(-\frac{\epsilon}{2}(\sqrt{d}+\sqrt{m})^2)$. If we let $R=\frac{m}{d}\to\infty$, then the above upper bound becomes 
\begin{equation}\label{relatedwork-Gauss2}
\beta_{\boldsymbol{A}}\leq (1+\epsilon)\cdot (1+o(1))\cdot e\cdot 2\sqrt{2}\leq 7.689\cdot (1+\epsilon).
\end{equation}
In contrast, our estimate is $\beta_0^{\mathbb{R}}\leq\beta_{\boldsymbol{A}}\leq \beta_0^{\mathbb{R}}+\epsilon$, where $\beta_0^{\mathbb{R}}=\sqrt{\frac{\pi}{\pi-2}}\approx 1.659$, provided that $m\geq C\log(1/\epsilon)\epsilon^{-2}d$ for a universal positive constant $C$. Therefore, our result improves the estimate in \eqref{relatedwork-Gauss2}.

\subsection{Notation}

Throughout this paper, we use the notation $\mathrm{i}$ to represent the imaginary unit, i.e., $\mathrm{i}=\sqrt{-1}$. For a complex number $a\in \mathbb{C}$, we use $\mathrm{Re}(a)$ and $\mathrm{Im}(a)$ to denote its real and imaginary parts, respectively.
We use $\| \boldsymbol{x}\|_2$ to denote the Euclidean norm of a vector $\boldsymbol{x}\in\mathbb{H}^d$, and we use $\| \boldsymbol{A}\|_2$ to denote the spectral norm of a matrix $\boldsymbol{A}\in\mathbb{H}^{m\times d}$. For a subset $I\subset[m]$ of size $k$, we use $\boldsymbol{A}_I\in\mathbb{H}^{k\times d}$ to denote the row submatrix of $\boldsymbol{A}$ consisting of rows whose indexes are in the subset $I$.
We use the notation $\mathbb{S}_{\mathbb{H}}^{d-1}$ to denote the unit sphere in $\mathbb{H}^d$, i.e., $$\mathbb{S}_{\mathbb{H}}^{d-1}=\{\boldsymbol{x}\in \mathbb{H}^d: \| \boldsymbol{x}\|_2=1\}.$$ For convenience, if $\mathbb{H}=\mathbb{R}$ then we simply write $\mathbb{S}_{\mathbb{R}}^{d-1}$ as $\mathbb{S}^{d-1}$.
We say a vector $\boldsymbol{a}\in\mathbb{H}^d$ is a standard Gaussian random vector if
\begin{equation*}
\begin{cases}
\boldsymbol{a}\sim \mathcal{N}(\boldsymbol{0},\boldsymbol{I}), &\text{if}\ \mathbb{H}=\mathbb{R};\\
\boldsymbol{a}\sim \mathcal{N}(\boldsymbol{0},\boldsymbol{I}/2)+\mathrm{i} \mathcal{N}(\boldsymbol{0},\boldsymbol{I}/2), &\text{if}\ \mathbb{H}=\mathbb{C}.
\end{cases}
\end{equation*}
We say a matrix $\boldsymbol{A}\in\mathbb{H}^{m\times d}$ is a standard Gaussian random matrix if the rows of $\boldsymbol{A}$ are i.i.d. standard Gaussian random vectors.

For any $A,B\in\mathbb{R}$, we use $A\gtrsim B$ to denote $A\geq C_0\cdot B$ where $C_0>0$ is an absolute constant. We define the notion $\lesssim$ in a similar way. 
Throughout this paper, we use $C$ and $c$, along with their subscripts or superscripts, to denote universal constants that may vary depending on the specific context.

\section{A universal lower bound for $\beta_{\boldsymbol{A}}$}

The aim of this section is to present a universal lower bound on $\beta_{\boldsymbol{A}}$ for $\boldsymbol{A}\in \mathbb{H}^{m\times d}$, where $\mathbb{H}=\mathbb{R}$ or $\mathbb{C}$. We need the following theorem, which shows that the optimal upper Lipschitz bound  $U_{\boldsymbol{A}}$ is equal to the spectral norm of ${\boldsymbol{A}}$ for both the real and complex cases. The real case of Theorem \ref{th:LUbound} was proved in \cite{BCMN,BW}, and the complex case was proved in \cite{AG}.
\begin{theorem}\label{th:LUbound}
\cite{AG,BCMN,BW}
Let $\boldsymbol{A}\in \mathbb{H}^{m\times d}$, where $\mathbb{H}=\mathbb{R}$ or $\mathbb{C}$. Then $U_{\boldsymbol{A}}=\|\boldsymbol{A}\|_2$.
\end{theorem}

Our main result of this section is Theorem \ref{lowerbound-ndim}, which presents a universal lower bound on the condition number $\beta_{\boldsymbol{A}}=\frac{U_{\boldsymbol{A}}}{L_{\boldsymbol{A}}}$ for all $\boldsymbol{A}\in\mathbb{H}^{m\times d}$, where $\mathbb{H}=\mathbb{R}$ or $\mathbb{C}$.
In order to obtain a lower bound on the condition number, we separately estimate the optimal Lipschitz bounds $U_{\boldsymbol{A}}$ and $L_{\boldsymbol{A}}$.
To establish a lower bound on $U_{\boldsymbol{A}}$, we utilize Theorem \ref{th:LUbound} and estimate the spectral norm of $\boldsymbol{A}$.
To obtain an upper bound on $L_{\boldsymbol{A}}$, we estimate the minimum possible ratio between $\|\abs{\boldsymbol{A}\vx}-\abs{\boldsymbol{A}\vy}\|_2$ and $\dist(\vx,\vy)$ when the pair $(\vx,\vy)$ belongs to a carefully chosen subset of $\{(\boldsymbol{z},\boldsymbol{w}):\langle \boldsymbol{z},\boldsymbol{w}\rangle=0\}$.

\begin{theorem}\label{lowerbound-ndim}
Let $\boldsymbol{A}=(\boldsymbol{a}_1,\ldots,\boldsymbol{a}_m)^*\in\mathbb{H}^{m\times d}$, where $\mathbb{H}=\mathbb{R}$ or $\mathbb{C}$. 
Then we have
\begin{equation}\label{eq:mdlower}
\beta_{\boldsymbol{A}}\geq \beta_0^{\mathbb{H}}=\begin{cases}
\sqrt{\frac{\pi}{\pi-2}}\,\,\approx\,\, 1.659  & \text{if $\mathbb{H}=\mathbb{R}$,}\\
\sqrt{\frac{4}{4-\pi}}\,\,\approx\,\, 2.159  & \text{if $\mathbb{H}=\mathbb{C}$.}\\

\end{cases}
\end{equation} 	

\end{theorem}

\begin{proof}
First, we make the assumption that (\ref{eq:mdlower}) holds for $d=2$. We will now prove that it holds for any $d> 2$.
Let $\boldsymbol{B}=(\boldsymbol{b}_1,\ldots,\boldsymbol{b}_m)^*\in\mathbb{H}^{m\times 2}$ be the matrix consisting of the first two columns of $\boldsymbol{A}\in \mathbb{H}^{m\times d}$.
According to the definition of $L_{\boldsymbol{A}}$ and $U_{\boldsymbol{A}}$, we have
\begin{equation*}
\begin{aligned}
L_{\boldsymbol{A}}&= \inf_{\substack{\boldsymbol{x},\boldsymbol{y}\in \mathbb{H}^d\\ \dist{(\boldsymbol{x},\boldsymbol{y})}\neq 0}}\frac{\|\abs{\boldsymbol{A}\boldsymbol{x}}-\abs{\boldsymbol{A}\boldsymbol{y}}\|_2}{\dist(\boldsymbol{x},\boldsymbol{y})}
\leq \inf_{\substack{\boldsymbol{x},\boldsymbol{y}\in \mathbb{H}^2\\ \dist{(\boldsymbol{x},\boldsymbol{y})}\neq 0}}\frac{\|\abs{\boldsymbol{B}\boldsymbol{x}}-\abs{\boldsymbol{B}\boldsymbol{y}}\|_2}{\dist(\boldsymbol{x},\boldsymbol{y})}=L_{\boldsymbol{B}}\\
\end{aligned} 
\end{equation*}
and
\begin{equation*}
\begin{aligned}
U_{\boldsymbol{A}}&= \sup_{\substack{\boldsymbol{x},\boldsymbol{y}\in \mathbb{H}^d\\ \dist{(\boldsymbol{x},\boldsymbol{y})}\neq 0}}\frac{\|\abs{\boldsymbol{A}\boldsymbol{x}}-\abs{\boldsymbol{A}\boldsymbol{y}}\|_2}{\dist(\boldsymbol{x},\boldsymbol{y})}
\geq \sup_{\substack{\boldsymbol{x},\boldsymbol{y}\in \mathbb{H}^2\\ \dist{(\boldsymbol{x},\boldsymbol{y})}\neq 0}}\frac{\|\abs{\boldsymbol{B}\boldsymbol{x}}-\abs{\boldsymbol{B}\boldsymbol{y}}\|_2}{\dist(\boldsymbol{x},\boldsymbol{y})}=U_{\boldsymbol{B}}.
\end{aligned} 
\end{equation*}
Therefore, we have
\begin{equation*}\label{rk6}
\beta_{\boldsymbol{A}}=\frac{U_{\boldsymbol{A}}}{L_{\boldsymbol{A}}}\geq \frac{U_{\boldsymbol{B}}}{L_{\boldsymbol{B}}}=\beta_{\boldsymbol{B}}\geq \beta_0^{\mathbb{H}}.
\end{equation*}
This completes the proof.

It remains to prove  (\ref{eq:mdlower}) holds for $d=2$, i.e., $\boldsymbol{A}=(\boldsymbol{a}_1, \boldsymbol{a}_2,\ldots,\boldsymbol{a}_m)^*\in \mathbb{H}^{m\times 2}$. 
Without loss of generality, for each $1\leq i\leq m$ we write 
\begin{equation*}
\boldsymbol{a}_i=t_i\begin{pmatrix}
\cos\phi_i\cos\alpha_i\\
\sin\phi_i\cos\beta_i	
\end{pmatrix}
+t_i\cdot \mathrm{i} \cdot \begin{pmatrix}
\cos\phi_i\sin\alpha_i\\
\sin\phi_i\sin\beta_i	
\end{pmatrix},
\end{equation*}
where $t_i=\Vert \boldsymbol{a}_i\Vert_2\geq0$, $(\phi_i,\alpha_i,\beta_i)\in I^{\mathbb{H}}$ and
\begin{equation*}
I^{\mathbb{H}}:=\begin{cases}
[0,\pi]\times \{0\} \times \{0\} & \text{if $\mathbb{H}=\mathbb{R}$,}\\	
[0,\pi]\times [0,2\pi] \times [0,2\pi] & \text{if $\mathbb{H}=\mathbb{C}$.}\\	
\end{cases}	
\end{equation*}
Since $\boldsymbol{A}^*\boldsymbol{A}$ is a $2\times 2$ matrix, we have
\begin{equation*}
\Vert\boldsymbol{A}\Vert_2^2=\Vert\boldsymbol{A}^*\boldsymbol{A}\Vert_2\geq \frac{1}{2}\cdot \mathrm{Tr}(\boldsymbol{A}^*\boldsymbol{A})=\frac{1}{2}\cdot \mathrm{Tr}(\boldsymbol{A}\boldsymbol{A}^*)=\frac{1}{2}\sum_{i=1}^{m}\Vert\boldsymbol{a}_i\Vert_2^2=\frac{1}{2}\sum_{i=1}^{m}t_i^2.
\end{equation*}
Also note that $U_{\boldsymbol{A}}=\Vert\boldsymbol{A}\Vert_2$ and 
\begin{equation}\label{xH2024}
\begin{aligned}
(L_{\boldsymbol{A}})^2&=\inf_{\boldsymbol{x},\boldsymbol{y}\in \mathbb{H}^{2},\mathrm{dist}(\boldsymbol{x},\boldsymbol{y})\neq 0}\frac{\||\boldsymbol{A}\boldsymbol{x}|-|\boldsymbol{A}\boldsymbol{y}|\|_2^2}{\mathrm{dist}^2(\boldsymbol{x},\boldsymbol{y})}
&\leq  \min_{(\boldsymbol{x},\boldsymbol{y})\in \mathcal{X}^{\mathbb{H}}}\frac{\||\boldsymbol{A}\boldsymbol{x}|-|\boldsymbol{A}\boldsymbol{y}|\|_2^2}{\mathrm{dist}^2(\boldsymbol{x},\boldsymbol{y})}=:M_{\boldsymbol{A}},
\end{aligned}
\end{equation} 
where $\mathcal{X}^{\mathbb{H}}\subset \mathbb{H}^2\times \mathbb{H}^2$ is defined as
\begingroup\fontsize{9pt}{12pt}\selectfont
\begin{equation*}
\mathcal{X}^{\mathbb{H}}:=\left\{ (\boldsymbol{x},\boldsymbol{y})\ \Bigg|\ \boldsymbol{x}=\begin{pmatrix}
\cos\theta\cos\alpha\\
\sin\theta\cos\beta	
\end{pmatrix}+\mathrm{i} 
\begin{pmatrix}
\cos\theta\sin\alpha\\
\sin\theta\sin\beta	
\end{pmatrix},
\boldsymbol{y}=\begin{pmatrix}
\sin\theta\cos\alpha\\
-\cos\theta\cos\beta	
\end{pmatrix}+\mathrm{i}  
\begin{pmatrix}
\sin\theta\sin\alpha\\
-\cos\theta\sin\beta	
\end{pmatrix},
(\theta,\alpha,\beta)\in I^{\mathbb{H}}
\right\}	.
\end{equation*}
\endgroup
Therefore, we have
\begin{equation}\label{rk7}
\beta_{\boldsymbol{A}}\,\,=\,\,\frac{U_{\boldsymbol{A}}}{L_{\boldsymbol{A}}}\,\,=\,\,\frac{\Vert\boldsymbol{A}\Vert_2}{L_{\boldsymbol{A}}}\geq\sqrt{\frac{\sum_{i=1}^{m}t_i^2}{2\cdot  M_{\boldsymbol{A}}}}.	
\end{equation} 
Then, to prove $\beta_{\boldsymbol{A}}\geq \beta_0^{\mathbb{H}}$, it is enough to show that 
\begin{equation}\label{eqxu:2}
M_{\boldsymbol{A}}\leq \frac{1}{2\cdot (\beta_0^{\mathbb{H}})^2}\sum_{i=1}^{m}t_i^2.
\end{equation}
A simple calculation shows that for any $(\theta,\alpha,\beta)\in I^{\mathbb{H}}$, by letting
\begin{equation*}
\boldsymbol{x}=\begin{pmatrix}
\cos\theta\cos\alpha\\
\sin\theta\cos\beta	
\end{pmatrix}+\mathrm{i} 
\begin{pmatrix}
\cos\theta\sin\alpha\\
\sin\theta\sin\beta	
\end{pmatrix}\quad\text{and}\quad 
\boldsymbol{y}=\begin{pmatrix}
\sin\theta\cos\alpha\\
-\cos\theta\cos\beta	
\end{pmatrix}+\mathrm{i} 
\begin{pmatrix}
\sin\theta\sin\alpha\\
-\cos\theta\sin\beta	
\end{pmatrix},
\end{equation*}
we have $(\boldsymbol{x},\boldsymbol{y})\in \mathcal{X}^{\mathbb{H}}$, $\langle \boldsymbol{x},\boldsymbol{y}\rangle=0$, $\mathrm{dist}^2(\boldsymbol{x},\boldsymbol{y})=\Vert\boldsymbol{x}-\boldsymbol{y}\Vert_2^2=2$, and
\begin{equation*}
\begin{aligned}
|\boldsymbol{a}_i^*\boldsymbol{x}|^2&= \frac{1}{2}t_i^2+\frac{1}{2}t_i^2\Big(\cos2\phi_i\cos2\theta+\sin2\phi_i\sin2\theta\cos(\alpha-\beta-\alpha_i+\beta_i)\Big),\\
|\boldsymbol{a}_i^*\boldsymbol{y}|^2&= \frac{1}{2}t_i^2-\frac{1}{2}t_i^2\Big(\cos2\phi_i\cos2\theta+\sin2\phi_i\sin2\theta\cos(\alpha-\beta-\alpha_i+\beta_i)\Big).
\end{aligned}
\end{equation*}
Therefore, we have
\begin{equation}\label{eqxu12:2}
\begin{aligned}
M_{\boldsymbol{A}}&=\min_{(\boldsymbol{x},\boldsymbol{y})\in \mathcal{X}^{\mathbb{H}}}\frac{\||\boldsymbol{A}\boldsymbol{x}|-|\boldsymbol{A}\boldsymbol{y}|\|_2^2}{\mathrm{dist}^2(\boldsymbol{x},\boldsymbol{y})}=\min_{(\boldsymbol{x},\boldsymbol{y})\in \mathcal{X}^{\mathbb{H}}}\frac{1}{2}\sum_{i=1}^m (|\boldsymbol{a}_i^*\boldsymbol{x}|^2+|\boldsymbol{a}_i^*\boldsymbol{y}|^2)- \sum_{i=1}^m |\boldsymbol{a}_i^*\boldsymbol{x}\|\boldsymbol{a}_i^*\boldsymbol{y}|\\
&=\min_{(\theta,\alpha,\beta)\in I^{\mathbb{H}}}\frac{1}{2}\sum_{i=1}^{m}t_i^2-f(\theta,\alpha,\beta)=\frac{1}{2}\sum_{i=1}^{m}t_i^2-\max_{(\theta,\alpha,\beta)\in I^{\mathbb{H}}}	f(\theta,\alpha,\beta),
\end{aligned}
\end{equation}
where $f(\theta,\alpha,\beta)$ is defined as
\begin{equation*}
f(\theta,\alpha,\beta):=\frac{1}{2}\sum_{i=1}^m t_i^2 \sqrt{ 1-\Big(\cos2\phi_i\cos2\theta+\sin2\phi_i\sin2\theta\cos(\alpha-\beta-\alpha_i+\beta_i)\Big)^2  }.	
\end{equation*}
We claim that for all $t_i\geq 0$ and  $(\phi_i,\alpha_i,\beta_i)\in I^{\mathbb{H}}$, $i=1,\ldots,m$,
\begin{equation}\label{eqxu12:1}
\max\limits_{(\theta,\alpha,\beta)\in I^{\mathbb{H}}}f(\theta,\alpha,\beta)\geq \Big(\frac{1}{2}-\frac{1}{2\cdot  (\beta_0^{\mathbb{H}})^2}\Big)\sum_{i=1}^{m}t_i^2= \begin{cases}
\frac{1}{\pi}\sum_{i=1}^m t_i^2	& \text{if $\mathbb{H}=\mathbb{R}$,}\\
\frac{\pi}{8}\sum_{i=1}^m t_i^2	& \text{if $\mathbb{H}=\mathbb{C}$.}
\end{cases}
\end{equation}
Then, plugging \eqref{eqxu12:1} into \eqref{eqxu12:2}, we immediately obtain \eqref{eqxu:2}.

It remains to prove \eqref{eqxu12:1}. We divide the proof into two cases.

\textbf{ Case 1: $\mathbb{H}=\mathbb{R}$.}
According to the definition of $I^{\mathbb{H}}$, it is enough to prove \eqref{eqxu12:1} when $\alpha=\beta=0$ and $\alpha_i=\beta_i=0$, $i=1,\ldots,m$. In this case, $f(\theta,0,0)$ can be simplified as
\begin{equation*}
f(\theta,0,0)=\frac{1}{2}\sum_{i=1}^m t_i^2 \sqrt{ 1-\Big(\cos2\phi_i\cos2\theta+\sin2\phi_i\sin2\theta\Big)^2 }
=\frac{1}{2}\sum_{i=1}^m t_i^2 |\sin(2\theta-2\phi_i)|.	
\end{equation*}
Since $f(\theta,0,0)$ is nonnegative for each $\theta$,
we have
\begin{equation}\label{eqxu:5}
\max\limits_{(\theta,\alpha,\beta)\in I^{\mathbb{R}}}f(\theta,\alpha,\beta)=\max\limits_{\theta\in [0,\pi]}f(\theta,0,0)\geq \frac{1}{\pi} \int_{0}^{\pi} f(\theta,0,0)\text{d}\theta=\frac{1}{\pi} \int_{0}^{\pi} \frac{1}{2} \sum_{i=1}^{m}t_i^2 |\sin(2\theta-2\phi_i)|\text{d}\theta.	
\end{equation}
A simple calculation shows that
\begin{equation}\label{eqxu:51}
\int_{0}^{\pi}\frac{1}{2} \sum_{i=1}^{m}t_i^2|\sin(2\theta-2\phi_i)|\text{d}\theta
=\sum_{i=1}^m \frac{t_i^2}{2}\int_{0}^{\pi}  |\sin(2\theta-2\phi_i)|\text{d}\theta\overset{(a)}=\sum_{i=1}^m t_i^2\int_{\phi_i}^{\phi_i+\frac{\pi}{2}}  \sin(2\theta-2\phi_i)\text{d}\theta=\sum_{i=1}^m t_i^2,	
\end{equation}
where ($a$) follows from the fact that $|\sin(2\theta-2\phi_i)|$ is a periodic function in $\theta$ with period $\frac{\pi}{2}$. Substituting equation \eqref{eqxu:51} into \eqref{eqxu:5}, we arrive at \eqref{eqxu12:1} when $\mathbb{H}=\mathbb{R}$.

\textbf{ Case 2: $\mathbb{H}=\mathbb{C}$.} 
Since $\cos(\alpha-\beta-\alpha_i+\beta_i)=\cos(\alpha-\beta)\cos(\alpha_i-\beta_i)+\sin(\alpha-\beta)\sin(\alpha_i-\beta_i)$ for each $1\leq i\leq m$, we can rewrite $f(\theta,\alpha,\beta)$ as
\begin{equation*}
f(\theta,\alpha,\beta)=\frac{1}{2}	\sum_{i=1}^m t_i^2  \sqrt{ 1-(a_i\cdot x+b_i\cdot y+c_i\cdot z)^2  }=:h(x,y,z),
\end{equation*}
where $x=\cos2\theta$, $y=\sin2\theta\cos(\alpha-\beta)$, $z=\sin2\theta\sin(\alpha-\beta)$, $a_i=\cos2\phi_i$, $b_i=\sin2\phi_i\cos(\alpha_i-\beta_i)$, $c_i=\sin2\phi_i\sin(\alpha_i-\beta_i)$.
Note that $x^2+y^2+z^2=1$ and $a_i^2+b_i^2+c_i^2=1$, $i=1,\ldots,m$. 
Therefore, proving \eqref{eqxu12:1} for the case $\mathbb{H}=\mathbb{C}$ is equivalent to proving that for all $(a_i,b_i,c_i)\in\mathbb{S}^2$, $i=1,\ldots,m$, we have
\begin{equation}\label{eqxx4}
\max\limits_{(x,y,z)\in \mathbb{S}^{2} }	h(x,y,z)\geq  \frac{\pi}{8} \sum_{i=1}^{m}t_i^2.
\end{equation}	
Here, $\mathbb{S}^{2}$ denotes the set of unit-norm vectors in $\mathbb{R}^3$. Note that $h(x,y,z)$ is nonnegative for all $(x,y,z)\in \mathbb{S}^{2}$, so we have
\begin{equation}\label{eqxx3}
\max\limits_{(x,y,z)\in \mathbb{S}^{2} }	h(x,y,z)\geq \frac{1}{4\pi\cdot 1^2} \iint_{\mathbb{S}^{2}}	 h(x,y,z) \text{d}S=\frac{1}{8\pi}	\sum_{i=1}^m t_i^2 \iint_{\mathbb{S}^{2}}	   \sqrt{ 1-(a_i\cdot x+b_i\cdot y+c_i\cdot z)^2  } \text{d}S.
\end{equation}
According to Poisson formula in surface integrals, for any continuous univariate function $p(u)$ and for any real numbers $a,b,c$, we have
\begin{equation}\label{Poisson formula}
\iint_{\mathbb{S}^{2}}	  p(a\cdot x+b\cdot y+c\cdot z)  \text{d}S
=2\pi  	\int_{-1}^{1}	  p(u\cdot \sqrt{a^2+b^2+c^2} )  \text{d}u.
\end{equation}
Substituting $a=a_i$, $b=b_i$, $c=c_i$, $p(u)=\sqrt{1-u^2}$ into \eqref{Poisson formula} and using the fact that $a_i^2+b_i^2+c_i^2=1$, we obtain
\begin{equation*}
\iint_{\mathbb{S}^{2}}	   \sqrt{ 1-(a_i\cdot x+b_i\cdot y+c_i\cdot z)^2  } \text{d}S	
=	2\pi  	\int_{-1}^{1}	  \sqrt{1-u^2}  \text{d}u=2\pi\cdot \frac{1}{2}\pi \cdot 1^2=\pi^2.
\end{equation*}
Substituting the above equation into \eqref{eqxx3}, we arrive at \eqref{eqxx4}. This completes the proof.

\end{proof}

\section{Estimate the condition number $\beta_{\boldsymbol{A}}$ for $\boldsymbol{A}\in\mathbb{R}^{m\times 2}$}\label{section:dim2}

In this section we focus on the case $d=2$ and give a deeper investigation on the condition number $\beta_{\boldsymbol{A}}$ where $\boldsymbol{A}\in\mathbb{H}^{m\times 2}$. We mainly focus on the real case $\mathbb{H}=\mathbb{R}$, and our objective is to determine the matrix $\boldsymbol{A}\in \mathbb{R}^{m\times 2} $ that minimizes $\beta_{\boldsymbol{A}}$ among all $m\times 2$ real matrices. 
The findings presented in this section are anticipated to offer significant insights into answering Question II, i.e., identifying the optimal vector set for phase retrieval. 

Throughout this section we denote the harmonic frame in $\R^2$ as $\boldsymbol{E}_m$. In other words, the rows of $\boldsymbol{E}_m$ consists of $m$ equidistant points on the semicircle, i.e., 
\begin{equation}\label{equidistant}
\boldsymbol{E}_m:=
\begingroup
\arraycolsep=5pt\def\arraystretch{1.4}
\begin{pmatrix}
1 &  \cos(\frac{1}{m}\pi) & \cdots & \cos(\frac{m-1}{m}\pi)  \\
0 & \sin(\frac{1}{m}\pi) &\cdots &  \sin(\frac{m-1}{m}\pi) \\
\end{pmatrix}
\endgroup
^T\,\,\in\,\,\R^{m\times 2}.
\end{equation}

We now present the central outcome of this section, which establishes that for any odd integer $m\geq 3$, a collection of $m$ equidistant points on the semicircle attains the minimum condition number $\beta_{\boldsymbol{A}}$ for $\boldsymbol{A}\in \R^{m\times 2}$.

\begin{theorem}\label{th:opti}
Let $m\geq 3$ be an odd integer, and let $\boldsymbol{E}_m$ be defined as in \eqref{equidistant}.
Then, 
\begin{equation*}
\boldsymbol{E}_m\,\,\in\,\, \argmin{\boldsymbol{A}\in \R^{m\times 2}} \beta_{\boldsymbol{A}},
\end{equation*}
i.e., $\beta_{\boldsymbol{E}_m}=\min_{\boldsymbol{A}\in \R^{m\times 2}} \beta_{\boldsymbol{A}}$.
\end{theorem}

In order to establish the proof of Theorem \ref{th:opti}, we introduce the following two theorems.

\begin{theorem}\label{lowerbound-2dim}
Let $d\geq 2$, $m\geq 3$ and let $\boldsymbol{A}=(\boldsymbol{a}_1,\ldots,\boldsymbol{a}_m)^T\in\mathbb{R}^{m\times d}$.  Then we have
\begin{equation}	\label{eqxu:6}
\beta_{\boldsymbol{A}}\,\,\geq\,\,  \frac{1}{\sqrt{1-\frac{1}{m\cdot \sin\frac{\pi}{2m}}}}.
\end{equation} 	
\end{theorem}

\begin{theorem}\label{cos_lemma}
Let $m\geq 3$ and let $\boldsymbol{E}_m$ be defined as in \eqref{equidistant}.
Then
\begin{equation}\label{betaA-R2}
\beta_{\boldsymbol{E}_m}=\begin{cases}
\frac{1}{\sqrt{1-\frac{2}{m\cdot \sin\frac{\pi}{m}}}}  & \text{if $m$ is even,}\\	
\frac{1}{\sqrt{1-\frac{1}{m\cdot \sin\frac{\pi}{2m}}}} & \text{if $m$ is odd.}
\end{cases}	
\end{equation}

\end{theorem}

We next state a formal proof of Theorem \ref{th:opti}.
\begin{proof}[Proof of Theorem \ref{th:opti}]
The proof of Theorem \ref{th:opti} readily follows from the combination of Theorem \ref{lowerbound-2dim} and Theorem \ref{cos_lemma}. Specifically, Theorem \ref{lowerbound-2dim} establishes a  lower bound on $\beta_{\boldsymbol{A}}$ for any $\boldsymbol{A}\in \R^{m\times d}$. On the other hand, Theorem \ref{cos_lemma} provides the precise value of $\beta_{\boldsymbol{E}_m}$, which perfectly matches the lower bound stated in  Theorem \ref{lowerbound-2dim} when $d=2$ and $m\geq 3$ is an odd integer. 
 As a result, we achieve the intended outcome of Theorem \ref{th:opti}.
\end{proof}

\begin{remark}
 Theorem \ref{lowerbound-2dim} provides a  lower bound on $\beta_{\boldsymbol{A}}$ for all $\boldsymbol{A}\in \R^{m\times d}$, which slightly improves the constant lower bound $\sqrt{\frac{\pi}{\pi-2}}$ in \eqref{eq:mdlower}.
 Moreover, Theorem \ref{cos_lemma} shows that $$\lim_{m\rightarrow \infty}{\beta_{\boldsymbol{E}_m}}\,\,=\,\,\sqrt{\frac{\pi}{\pi-2}},$$ which indicates that the constant lower bound $\sqrt{\frac{\pi}{\pi-2}}$ in \eqref{eq:mdlower} is asymptotic optimal  as $m\to\infty$. 
\end{remark}

Inspired by Theorem \ref{th:opti}, we propose the following conjecture for any even integer $m \geq 4$.

\begin{conjecture}\label{conj1}
Assume that $m\geq 4$ is an even integer. Then, 
\begin{equation*}
\boldsymbol{E}_m\,\,\in\,\, \argmin{\boldsymbol{A}\in \R^{m\times 2}} \beta_{\boldsymbol{A}}.
\end{equation*}

\end{conjecture}

\subsection{Proof of Theorem \ref{lowerbound-2dim}}

In this subsection we prove Theorem \ref{lowerbound-2dim}. We first prove the following lemma, which slightly improves the lower bound \eqref{eqxu12:1} for the case $\mathbb{H}=\mathbb{R}$.

\begin{lem}\label{beta_A:d2m3}
For any real numbers $\phi_1, \phi_2,\ldots,\phi_{m+1}$ satisfying $0=\phi_1\leq\phi_2\leq\cdots\leq \phi_m\leq \phi_{m+1}= \pi$, and for all $t_1,\ldots,t_m\in\mathbb{R}$, we have 
\begin{equation}\label{eqxu3:00}
\max_{\theta\in[0,\pi]}\sum_{i=1}^{m}t_i^2|\sin(\theta-\phi_i)|\geq \frac{1}{m\cdot \sin\frac{\pi}{2m}}\sum_{i=1}^{m}t_i^2.
\end{equation}
\end{lem}

\begin{proof}
For convenience, we set $g(\theta):=	\sum_{i=1}^{m}t_i^2|\sin(\theta-\phi_i)|$, and  we define
\begin{equation}\label{rk3}
g_k(\theta):=\sum_{i=1}^{k}t_i^2\sin(\theta-\phi_i) -\sum_{i=k+1}^{m}t_i^2\sin(\theta-\phi_i)=r_k\sin(\theta-\theta_k)
\end{equation}
for each $1\leq k\leq m$, where $\theta_k\in[0,\pi]$ and $r_k\geq 0$ satisfy that
\begin{equation*}
\cos\theta_k= \frac{1}{r_k}\cdot \Big(\sum_{i=1}^{k}t_i^2\cos\phi_i-\sum_{i= k+1}^{m}t_i^2\cos\phi_i\Big),\quad \sin\theta_k= \frac{1}{r_k}\cdot \Big(\sum_{i=1}^{k}t_i^2\sin\phi_i-\sum_{i= k+1}^{m}t_i^2\sin\phi_i\Big)
\end{equation*}
and
\begin{equation*}
r_k^2=	\Big(\sum_{i=1}^{k}t_i^2\cos\phi_i -\sum_{i= k+1}^{m}t_i^2\cos\phi_i\Big)^2+\Big(\sum_{i= 1}^{k}t_i^2\sin\phi_i -\sum_{i= k+1}^{m}t_i^2\sin\phi_i\Big)^2.
\end{equation*}
For each $k=1,\ldots,m$, note that $g(\theta)=g_k(\theta)$ when $\theta\in [\phi_k,\phi_{k+1}]$, so we have 
\begin{equation*}
\begin{aligned}
\int_{\phi_k}^{\phi_{k+1}}	g(\theta)\text{d}\theta&=\int_{\phi_k}^{\phi_{k+1}}	r_k\sin(\theta-\theta_k)\text{d}\theta=r_k\cdot \Big(\cos(\phi_k-\theta_k)-\cos(\phi_{k+1}-\theta_k)   \Big)\\
&=2\cdot r_k\cdot \sin(\frac{\phi_k+\phi_{k+1}}{2}-\theta_k)\sin(\frac{\phi_{k+1}-\phi_k}{2})\leq 2\cdot r_k\cdot \sin(\frac{\phi_{k+1}-\phi_k}{2}).
\end{aligned}
\end{equation*}
We assert that, for each $k=1,\ldots,m$, the following inequality holds:
\begin{equation}\label{rk}
r_k\leq 	\max_{\theta\in[0,\pi]}g(\theta).
\end{equation}
Then we have
\begin{equation*}
\begin{aligned}
&\quad \sum_{i=1}^mt_i^2\overset{(a)}=\frac{1}{2}\int_{0}^{\pi}	g(\theta)\text{d}\theta=\frac{1}{2}\sum_{k=1}^{m}\int_{\phi_k}^{\phi_{k+1}}	g(\theta)\text{d}\theta\leq \sum_{k=1}^{m} r_k\cdot \sin(\frac{\phi_{k+1}-\phi_k}{2})\\
&\leq \max_{\theta\in[0,\pi]}g(\theta) \cdot   \sum_{k=1}^{m} \sin(\frac{\phi_{k+1}-\phi_k}{2})\overset{(b)}\leq 	\max_{\theta\in[0,\pi]}g(\theta) \cdot m\cdot   \sin\Big(\sum_{k=1}^{m} \frac{\phi_{k+1}-\phi_k}{2m}\Big)=m\sin(\frac{\pi}{2m})\cdot 	\max_{\theta\in[0,\pi]}g(\theta),
\end{aligned}
\end{equation*}
which  implies the desired result in \eqref{eqxu3:00}. Here, equation ($a$) follows from \eqref{eqxu:51}, and  inequality ($b$) follows from Jensen's inequality and the fact that $\sin x$ is a concave function on $[0,\pi]$. 

It remains to prove \eqref{rk}. Let $k\in\{1,\ldots,m\}$ be fixed. 
Note that $g(\theta) =g_k(\theta)= r_k\sin(\theta-\theta_k) \neq 0$ when $\theta \in (\phi_k, \phi_{k+1})$. Hence, two cases can be observed: either $g(\theta)$ is monotonic on the interval $[\phi_k, \phi_{k+1}]$, or there exists a $\theta_* \in [\phi_k, \phi_{k+1}]$ such that $g(\theta)$ is monotonically increasing on the subinterval $[\phi_k, \theta_*]$ and monotonically decreasing on the subinterval $[\theta_*, \phi_{k+1}]$.
If such $\theta_*$ exists, then we have
\begin{equation*}
r_k = g_k(\theta_*) = \max_{\theta\in[\phi_k,\phi_{k+1}]} g_k(\theta)= \max_{\theta\in[\phi_k,\phi_{k+1}]} g(\theta) \leq \max_{\theta\in[0,\pi]} g(\theta),
\end{equation*}
which directly leads to the desired result. 
Next, let us consider the case when $g(\theta)$ is monotonic on the interval $[\phi_k,\phi_{k+1}]$. We will divide the proof into two separate cases.

\textbf{ Case 1: $g(\theta)$ is monotonically increasing on $[\phi_k,\phi_{k+1}]$.}
For each $j=2,3,\ldots,k$, set $\phi_{j+m}=\pi+\phi_j$, $r_{j+m}=r_j$, $\theta_{j+m}=\theta_j$  and $g_{j+m}(\theta)=g_j(\theta)$, where $r_j$, $\theta_j$ and $g_j(\theta)$ are defined in \eqref{rk3}. Then we have $g(\theta)=g_j(\theta)$ for each $\theta\in[\phi_{j},\phi_{j+1}]$, $j=k,\ldots,k+m-1$.
Note that $\phi_{k+m}=\pi+\phi_{k}$, and that $g(\theta)$ is a periodic function on $\mathbb{R}$ with period $\pi$, it is enough to prove that 
\begin{equation}\label{rk2024}
r_k\leq 	\max_{\theta\in[\phi_k,\phi_{k+m}]}g(\theta).
\end{equation}

Since $g(\theta)$ is continuous and has period $\pi$, $g(\theta)$ can not be monotonically increasing on $[\phi_k,\phi_{k+m}]$. Let $l\in\{k+1,\ldots,m+k-1\}$ be the smallest integer $K$ such that $g(\theta)$ is not monotonically increasing on $[\phi_{K},\phi_{K+1}]$. We will prove \eqref{rk2024} by showing that
\begin{equation}\label{rk20242-10}
r_k\leq r_l\leq \max_{\theta\in[\phi_k,\phi_{k+m}]}g(\theta).
\end{equation}

For each $j=k,k+1,\ldots,l-1$, since $g(\theta)=g_j(\theta)$ is monotonically increasing on $[\phi_{j},\phi_{j+1}]$, we have $g_j'(\phi_{j+1})\geq0$. 
By a direct calculation, for each $j\in\{k,k+1,\ldots,l-1\}$ we have
\begin{equation*}
r_{j+1}^2-r_j^2=	4t_{j+1}^2\cdot g_j'(\phi_{j+1})+4t_{j+1}^4\geq 0\quad\text{and}\quad g_{j+1}'(\phi_{j+1})=g_{j}'(\phi_{j+1})+2t_{j+1}^2\geq0.
\end{equation*}
In particular, we have $g_{l}'(\phi_{l})\geq 0$ and
\begin{equation}\label{rk1}
r_k\leq r_{k+1}\leq \cdots\leq r_{l-1}\leq r_{l}.
\end{equation}
Recall that $g_{l}(\theta)=r_{l}\sin(\theta-\theta_{l})$ is not monotonically increasing on $[\phi_{l},\phi_{l+1}]$, so the non-negativity of $g_{l}'(\phi_{l})$ guarantees the existence of a $\phi_*\in[\phi_{l},\phi_{l+1}]$, such that $g_{l}(\theta)$ is monotonically increasing on $[\phi_{l},\phi_*]$ and monotonically decreasing on $[\phi_*,\phi_{l+1}]$. This means that
\begin{equation}\label{rk2}
r_{l}=g_{l}(\phi_*)=\max_{\theta\in[\phi_{l},\phi_{l+1}]}g_{l}(\theta)=\max_{\theta\in[\phi_{l},\phi_{l+1}]}g(\theta)\leq\max_{\theta\in[\phi_{k},\phi_{k+m}]}g(\theta).	
\end{equation}
Combining  \eqref{rk1} with \eqref{rk2}, we arrive at \eqref{rk20242-10}.

\textbf{ Case 2: $g(\theta)$ is monotonically decreasing on $[\phi_k,\phi_{k+1}]$.} The analysis is similar with Case 1. 
For each $j=1,2,\ldots,m-k$, set $\phi_{-j}=-\pi+\phi_{m+1-j}$, $r_{-j}=r_{m+1-j}$, $\theta_{-j}=\theta_{m+1-j}$  and $g_{-j}(\theta)=g_{m+1-j}(\theta)$. 
Note that $\phi_{-(m-k)}=-\pi+\phi_{k+1}$.
The continuity and periodicity of $g(\theta)$ guarantee that $g(\theta)$ is not monotonically decreasing on $[-\pi+\phi_{k+1},\phi_{k+1}]$. Let $l\in\{-(m-k),\ldots,k-1\}$ be the largest integer $K$ such that $g(\theta)$ is not monotonically decreasing on $[\phi_{K},\phi_{K+1}]$. Using a similar analysis as in Case 1, it can be shown that
\begin{equation*}
r_k\leq r_{k-1}\leq\cdots\leq r_{l+1}\leq  r_l=\max_{\theta\in[\phi_{l},\phi_{l+1}]}g_{l}(\theta)=\max_{\theta\in[\phi_{l},\phi_{l+1}]}g(\theta)\leq \max_{\theta\in[-\pi+\phi_{k+1},\phi_{k+1}]}g(\theta)=\max_{\theta\in[0,\pi]}g(\theta).
\end{equation*}
This completes the proof of \eqref{rk}.

\end{proof}

Now we can present a proof of Theorem \ref{lowerbound-2dim}.

\begin{proof}[Proof of Theorem \ref{lowerbound-2dim}]
According to the analysis of Theorem \ref{lowerbound-ndim}, it is enough to prove the theorem  for $d=2$.
Note that $\beta_{\boldsymbol{A}}$ is invariant if we switch any $\boldsymbol{a}_i$ to $-\boldsymbol{a}_i$. Therefore, without loss of generality, we assume that $\boldsymbol{a}_i=(t_i\cos\phi_i,t_i\sin\phi_i)^T$, $i=1,\ldots,m$, where $t_i=\Vert\boldsymbol{a}_i\Vert_2$ and $0= \phi_1\leq \phi_2\leq \cdots\leq \phi_m\leq \frac{\pi}{2}$. Recall that in Theorem \ref{lowerbound-ndim} we obtain the following lower bound on $\beta_{\boldsymbol{A}}$ (see equation \eqref{rk7}):
\begin{equation}\label{eq2024-1}
\beta_{\boldsymbol{A}}\geq 	\sqrt{\frac{\sum_{i=1}^{m}t_i^2}{2\cdot M_{\boldsymbol{A}}}},	
\end{equation}
where
\begin{equation*}
M_{\boldsymbol{A}}=\min_{(\boldsymbol{x},\boldsymbol{y})\in \mathcal{X}^{\mathbb{R}}}\frac{\||\boldsymbol{A}\boldsymbol{x}|-|\boldsymbol{A}\boldsymbol{y}|\|_2^2}{\mathrm{dist}^2(\boldsymbol{x},\boldsymbol{y})}\overset{(a)}=\frac{1}{2}\sum_{i=1}^{m}t_i^2-\max_{\theta\in [0,\pi]}\frac{1}{2}\sum_{i=1}^m t_i^2 |\sin(2\theta-2\phi_i)|.
\end{equation*}
Here, $\mathcal{X}^{\mathbb{R}}$ is defined in \eqref{xH2024}, and equality  ($a$) follows from \eqref{eqxu12:2}.
By Lemma \ref{beta_A:d2m3}, we have
\begin{equation}\label{ea2024}
M_{\boldsymbol{A}}\leq \frac{1}{2}\sum_{i=1}^{m}t_i^2-\frac{1}{2m\cdot \sin\frac{\pi}{2m}}\sum_{i=1}^{m}t_i^2 .
\end{equation}
Substituting \eqref{ea2024} into \eqref{eq2024-1}, we obtain 
\begin{equation*}
\beta_{\boldsymbol{A}}\,\,\geq\,\,  \frac{1}{\sqrt{1-\frac{1}{m\cdot \sin(\frac{\pi}{2m})}}},	
\end{equation*}
as desired.

\end{proof}

\subsection{Proof of Theorem \ref{cos_lemma}}

In this subsection, we provide a proof for Theorem \ref{cos_lemma}. To support our argument, we begin by introducing two lemmas that prove to be useful.

\begin{lem}\cite[Theorem 1.1]{AAFG24}\label{upperbound-orth}
Let $\boldsymbol{A}=(\boldsymbol{a}_1,\ldots,\boldsymbol{a}_m)^*\in\mathbb{H}^{m\times d}$, where $\mathbb{H}=\mathbb{R}$ or $\mathbb{C}$. 
Then we have
\begin{equation*}
L_{\boldsymbol{A}}=\min_{\substack{\boldsymbol{x}\in \mathbb{H}^{d},\boldsymbol{y}\in \mathbb{H}^{d} \\ \|\boldsymbol{x}\|_2=1, \|\boldsymbol{y}\|_2\leq 1,\langle \boldsymbol{x},\boldsymbol{y}\rangle=0  }}
\frac{\||\boldsymbol{A}\boldsymbol{x}|-|\boldsymbol{A}\boldsymbol{y}|\|_2}
{\mathrm{dist}(\boldsymbol{x},\boldsymbol{y})}.
\end{equation*} 	

\end{lem}

\begin{lem}\label{prop-calculation}
Let $m\geq 3$ be a positive integer. 
For any $\theta\in\mathbb{R}$, set
\begin{equation}\label{eq1:6}
G_m(\theta):=\sum_{j=0}^{m-1}\left|\sin\left(\frac{2j\pi}{m}+2\theta\right)\right|.
\end{equation}
Then we have
\begin{equation*}
\max\limits_{\theta\in[0,\pi]}G_m(\theta)=
\begin{cases}
\frac{2}{\sin\frac{\pi}{m}} & \text{if $m$ is even,}\\	
\frac{1}{\sin\frac{\pi}{2m}} & \text{if $m$ is odd.}
\end{cases} 	
\end{equation*}
In the case where $m$ is an even number, the equality is achieved if $\theta=\frac{\pi}{2m}$. On the other hand, if $m$ is odd, the equality is attained if $\theta=\frac{\pi}{4m}$.

\end{lem}
\begin{proof}
The proof is presented in Appendix \ref{Proof of prop-calculation}.
\end{proof}

Now we can present a proof of Theorem \ref{cos_lemma}.

\begin{proof}[Proof of Theorem \ref{cos_lemma}]
For convenience, we denote $\boldsymbol{A}=\boldsymbol{E}_m$ and denote the $i$-th row of $\boldsymbol{A}$ by $\boldsymbol{a}_i$.
Note that $\boldsymbol{A}^T\boldsymbol{A}=\frac{m}{2}\boldsymbol{I}$, so the upper Lipschitz bound $U_{\boldsymbol{A}}=\sqrt{\frac{m}{2}}$. 
Based on Lemma \ref{upperbound-orth}, we claim that
\begin{equation}\label{LA-R2}
(L_{\boldsymbol{A}})^2=\min_{\substack{\boldsymbol{x}\in \mathbb{H}^{d},\boldsymbol{y}\in \mathbb{H}^{d} \\ \|\boldsymbol{x}\|_2=1, \|\boldsymbol{y}\|_2\leq 1,\langle \boldsymbol{x},\boldsymbol{y}\rangle=0  }}\frac{\||\boldsymbol{A}\boldsymbol{x}|-|\boldsymbol{A}\boldsymbol{y}|\|_2^2}{\mathrm{dist}^2(\boldsymbol{x},\boldsymbol{y})}
=\begin{cases}
\frac{m}{2}-\frac{1}{\sin\frac{\pi}{m}} & \text{if $m$ is even,}\\	
\frac{m}{2}-\frac{1}{2\sin\frac{\pi}{2m}} & \text{if $m$ is odd.}\end{cases}	
\end{equation}
Combining with $U_{\boldsymbol{A}}=\sqrt{\frac{m}{2}}$,  we arrive at \eqref{betaA-R2}.

It remains to prove (\ref{LA-R2}).  We first consider the case when $\boldsymbol{y}= \boldsymbol{0}$. In this case we have 
\begin{equation}\label{LA-R22}
\frac{\||\boldsymbol{A}\boldsymbol{x}|-|\boldsymbol{A}\boldsymbol{y}|\|_2^2}{\mathrm{dist}^2(\boldsymbol{x},\boldsymbol{y})}
=\frac{\|\boldsymbol{A}\boldsymbol{x}\|_2^2}{\|\boldsymbol{x}\|^2_2}
=\frac{m}{2}.	
\end{equation}
We next consider the case when $\boldsymbol{y}\neq \boldsymbol{0}$. Without loss of generality, we can assume that $\boldsymbol{x}= (\cos\theta,-\sin\theta)^{\top}$ and $ \boldsymbol{y}=\|\boldsymbol{y}\|_2\cdot (-\sin\theta,-\cos\theta)^{\top}$,
where $\theta\in[0,\pi)$. Therefore, we have $\mathrm{dist}^2(\boldsymbol{x},\boldsymbol{y})=\|\boldsymbol{x}-\boldsymbol{y}\|_2^2=1+\|\boldsymbol{y}\|_2^2$. 
Recalling that $\boldsymbol{A}^T\boldsymbol{A}=\frac{m}{2}\boldsymbol{I}$, we can proceed with a direct calculation:
\begin{equation*}
\||\boldsymbol{A}\boldsymbol{x}|-|\boldsymbol{A}\boldsymbol{y}|\|_2^2=\boldsymbol{x}^T\boldsymbol{A}^T\boldsymbol{A}\boldsymbol{x}+\boldsymbol{y}^T\boldsymbol{A}^T\boldsymbol{A}\boldsymbol{y}-2\sum_{j=1}^{m}| \boldsymbol{x}^T\boldsymbol{a}_j\boldsymbol{a}_j^T\boldsymbol{y}|=\frac{m}{2}(1+\|\boldsymbol{y}\|_2^2)-\|\boldsymbol{y}\|_2\cdot G_m(\theta),
\end{equation*}
where $G_m(\theta)$ is defined in \eqref{eq1:6}.
Then we have 
\begin{equation}\label{LA-R222}
\frac{\||\boldsymbol{A}\boldsymbol{x}|-|\boldsymbol{A}\boldsymbol{y}|\|_2^2}{\mathrm{dist}^2(\boldsymbol{x},\boldsymbol{y})}
=\frac{m}{2}-\frac{G_m(\theta)}{\frac{1}{\|\boldsymbol{y}\|_2}+\|\boldsymbol{y}\|_2}
\overset{(a)}\geq \frac{m}{2}-\frac{1}{2}G_m(\theta)
\overset{(b)}\geq \begin{cases}
\frac{m}{2}-\frac{1}{\sin\frac{\pi}{m}} & \text{if $m$ is even,}\\	
\frac{m}{2}-\frac{1}{2\sin\frac{\pi}{2m}} & \text{if $m$ is odd.}
\end{cases}
\end{equation}
Here, in inequality ($a$) we use Cauchy-Schwarz inequality, and the equality is obtained when $\|\boldsymbol{y}\|_2=1$.  Inequality ($b$) follows from Lemma \ref{prop-calculation}.
Therefore, combining \eqref{LA-R22} and \eqref{LA-R222}, we arrive at \eqref{LA-R2}. 
If $m$ is an even integer, then the equality in \eqref{LA-R2} is attained if $\boldsymbol{x}= (\cos\frac{\pi}{2m},-\sin\frac{\pi}{2m})^{\top}$ and $ \boldsymbol{y}=(-\sin\frac{\pi}{2m},-\cos \frac{\pi}{2m})^{\top}$. Otherwise, if $m$ is an odd integer, then the equality in \eqref{LA-R2} is attained if $\boldsymbol{x}= (\cos\frac{\pi}{4m},-\sin\frac{\pi}{4m})^{\top}$ and $ \boldsymbol{y}=(-\sin\frac{\pi}{4m},-\cos \frac{\pi}{4m})^{\top}$.

\end{proof}

\section{Estimate the condition number $\beta_{\boldsymbol{A}}$ for Gaussian random matrix}\label{section:Gauss}
In this section, we estimate $\beta_\BA$ for a standard Gaussian random matrix $\boldsymbol{A}\in \mathbb{H}^{m\times d}$, where $\mathbb{H}=\mathbb{R}$ or $\mathbb{C}$. More concretely, in the following statements, the rows of $\boldsymbol{A}$ are independently drawn from $\mathcal{N}(\boldsymbol{0},\boldsymbol{I})$ when $\mathbb{H}=\mathbb{R}$, and  the rows of $\boldsymbol{A}$ are independently drawn from $\mathcal{N}(\boldsymbol{0},\boldsymbol{I}/2)+\mathrm{i}\mathcal{N}(\boldsymbol{0},\boldsymbol{I}/2)$ when $\mathbb{H}=\mathbb{C}$.  

Recall that in Theorem \ref{lowerbound-ndim} we derive a universal lower bound on the condition number $\beta_{\boldsymbol{A}}$ for all matrices $\boldsymbol{A}\in\mathbb{H}^{m\times d}$:
\begin{equation}\label{eq:mdlower-new}
\beta_{\boldsymbol{A}}\geq \beta_0^{\mathbb{H}}=\begin{cases}
\sqrt{\frac{\pi}{\pi-2}}\,\,\approx\,\, 1.659  & \text{if $\mathbb{H}=\mathbb{R}$,}\\
\sqrt{\frac{4}{4-\pi}}\,\,\approx\,\, 2.159  & \text{if $\mathbb{H}=\mathbb{C}$.}\\

\end{cases}
\end{equation}

The main result of this section is the following theorem, which presents an estimate on the condition number $\beta_{\boldsymbol{A}}$ for a  standard Gaussian random matrix $\boldsymbol{A}\in \mathbb{H}^{m\times d}$, where $\mathbb{H}=\mathbb{R}$ or $\mathbb{C}$. 

\begin{theorem}\label{thm: complex_beta}
Assume that $\boldsymbol{A}\in \mathbb{H}^{m\times d}$ is a standard Gaussian random matrix, where $\mathbb{H}=\mathbb{R}$ or $\mathbb{C}$. 
For any $0<\delta<0.4$, with probability at least  $1-2\exp(-c\delta^2m)$, we have 
\begin{equation}\label{thm: complex_beta:2024-1}
\beta_0^{\mathbb{H}}\leq \beta_{\boldsymbol{A}}\leq \beta_0^{\mathbb{H}}+\delta,
\end{equation}
provided that $m\gtrsim \log(1/\delta)\delta^{-2}d$.  
Here, $c$ is a universal positive constant, and $\beta_0^{\mathbb{H}}$ is the constant defined in \eqref{eq:mdlower-new}.
\end{theorem}

 It is worth noting that, by letting $\delta\rightarrow0$ in \eqref{thm: complex_beta:2024-1}, the condition number $\beta_{\boldsymbol{A}}$ of a standard Gaussian matrix $\boldsymbol{A}\in \mathbb{H}^{m\times d}$ approaches the constant $\beta_0^{\mathbb{H}}$ in \eqref{eq:mdlower-new}. 
Therefore, Theorem \ref{thm: complex_beta} indicates that the universal lower bound $\beta_0^{\mathbb{H}}$ is asymptotic optimal for both the real case and the complex case.

To establish the validity of Theorem \ref{thm: complex_beta}, we introduce several auxiliary results that provide estimations for the values of $U_\BA$ and $L_\BA$. 

Recall that by Theorem \ref{th:LUbound} we have $U_\BA=\|\boldsymbol{A}\|_2$ for all $\boldsymbol{A}\in \mathbb{H}^{m\times d}$.
The following lemma provides an estimate on $\|\boldsymbol{A}\|_2$ for a standard Gaussian random matrix $\boldsymbol{A}\in \mathbb{H}^{m\times d}$, which immediately gives an upper bound on $U_\BA$.

\begin{lem}\label{complex_fund1}\cite{CSV13,CS}
Assume that $\boldsymbol{A}\in \mathbb{H}^{m\times d}$ is a standard Gaussian random matrix, where $\mathbb{H}=\mathbb{R}$ or $\mathbb{C}$. 
For any $0<\delta<1$,  with probability at least $1-2\exp(-c_1\delta^2m)$, we have
\begin{equation}\label{eqn: union_RIP0}
(1-\delta)\|\boldsymbol{x}\|_2^2\leq \frac{1}{m}\|\boldsymbol{A}\boldsymbol{x}\|_2^2\leq (1+\delta)\|\boldsymbol{x}\|_2^2, \quad\forall \boldsymbol{x}\in\mathbb{H}^d,
\end{equation}
provided that $m\gtrsim \delta^{-2}d$. Here, $c_1$ is a universal positive constant. 
\end{lem}

In the following theorem we estimate the lower Lipschitz bound $L_{\boldsymbol{A}}$ for a standard Gaussian random matrix $\boldsymbol{A}\in \mathbb{H}^{m\times d}$.

\begin{theorem}\label{lem: large_distance_2}
Assume that $\boldsymbol{A}=(\boldsymbol{a}_1,\ldots,\boldsymbol{a}_m)^*\in \mathbb{H}^{m\times d}$ is a standard Gaussian random matrix, where $\mathbb{H}=\mathbb{R}$ or $\mathbb{C}$. 
For any $0<\delta<0.05$, there exists a universal constant $c_2>0$, such that if $m\gtrsim \log(1/\delta)\delta^{-2}d$, then with probability at least $1-4\exp(-c_2\delta^2 m)$ it holds that
\begin{equation}\label{xulem2024-8}
\frac{1}{\sqrt{m}}\cdot L_{\boldsymbol{A}}
\geq \frac{1}{\beta_0^{\mathbb{H}}}-\delta.
\end{equation}
\end{theorem}

Utilizing Lemma \ref{complex_fund1} and Theorem \ref{lem: large_distance_2}, we can present a proof of Theorem \ref{thm: complex_beta}. 

\begin{proof}[Proof of Theorem \ref{thm: complex_beta}]
By Theorem \ref{th:LUbound} we have $U_\BA=\|\boldsymbol{A}\|_2$. Then Lemma \ref{complex_fund1} implies that for any $\delta_1\in(0,0.05)$, if $m\gtrsim \delta_1^{-2}d$ then
\begin{equation}\label{march-UA1}
U_{\BA}=\|\boldsymbol{A}\|_2\leq (1+\delta_1)\cdot \sqrt{m}
\end{equation}
holds with probability at least $1-2\exp(-c_1\delta_1^2m)$. Here, $c_1$ is the universal positive constant in Lemma \ref{complex_fund1}.
On the other hand, Theorem \ref{lem: large_distance_2} shows that if $m\gtrsim \log(1/\delta_1)\delta_1^{-2}d$, then with probability at least $1-4\exp(-c_2\delta_1^2 m)$,  
\begin{equation}\label{march-LA4}
L_{\boldsymbol{A}}
\geq (\frac{1}{\beta_0^{\mathbb{H}}}-\delta_1)\cdot \sqrt{m}.
\end{equation}
Here, $c_2$ is the universal positive constant in Theorem \ref{lem: large_distance_2}.
Then, combining \eqref{march-UA1} and \eqref{march-LA4} we obtain that, with probability at least $1-2\exp(-c''\delta_1^2m)$, 
\begin{equation*}
\beta_{\BA}=\frac{U_{\BA}}{L_{\BA}}
\leq \frac{1+\delta_1}{\frac{1}{\beta_0^{\mathbb{H}}}-\delta_1}
\overset{(a)}\leq \beta_0^{\mathbb{H}}+  8\delta_1,
\end{equation*} 
where $c''$ is a universal constant and we use $\delta_1<0.05$ in inequality ($a$).  Take $\delta=8\delta_1$. Combining with the lower bound in Theorem \ref{lowerbound-ndim}, we arrive at our conclusion.

\end{proof}

In the following subsections, we will utilize Theorem \ref{lem: large_distance_2}  to showcase the performance of quadratic models in phase retrieval. Furthermore, a detailed proof for Theorem \ref{lem: large_distance_2} will be provided.

\subsection{The performance of quadratic models in phase retrieval}
For a given $\boldsymbol{x}_0 \in \mathbb{H}^d$ and a noise vector $\boldsymbol{\eta} \in \mathbb{R}^m$, a commonly used approach for estimating $\boldsymbol{x}_0$ from $\abs{\boldsymbol{A}\boldsymbol{x}_0}+\boldsymbol{\eta}$ is through the quadratic model defined as follows:
\begin{equation}\label{eq:hatx}
\hat{\boldsymbol{x}}\,\,\in\,\, \argmin{\boldsymbol{x}\in {\mathbb H}^d}\|\abs{\boldsymbol{A}\boldsymbol{x}}-\abs{\boldsymbol{A}\boldsymbol{x}_0}-\boldsymbol{\eta}\|_2.
\end{equation}

The objective of this subsection is to introduce the subsequent corollary, which demonstrates the performance of the widely used quadratic model in phase retrieval.
\begin{coro}\label{co:budeng}
Assume that $\boldsymbol{A}=(\boldsymbol{a}_1,\ldots,\boldsymbol{a}_m)^*\in \mathbb{H}^{m\times d}$ is a standard Gaussian random matrix, where $\mathbb{H}=\mathbb{R}$ or $\mathbb{C}$. 
For any $0 < \delta < 0.05$, $\boldsymbol{x}_0 \in \mathbb{H}^d$, and $\boldsymbol{\eta} \in \R^m$, the following inequality holds with a probability of at least $1-2\exp(-c\delta^2 m)$:
\begin{equation}\label{eq:performance}
\dist(\hat{\boldsymbol{x}}, \boldsymbol{x}_0)\,\,\leq \,\,\frac{2\beta_0^{\mathbb H}}{1-\delta}\cdot \frac{\|\boldsymbol{\eta}\|_2}{\sqrt{m}}
\end{equation}
provided that $m\gtrsim \log(1/\delta)\delta^{-2}d$. Here, $\hat{\boldsymbol{x}}$ and $\beta_0^{\mathbb H}$ is defined in (\ref{eq:hatx}) and (\ref{eq:mdlower-new}), respectively, and $c$ is a universal positive constant. 
\end{coro}
\begin{proof}
Note that
\[
\|\abs{\boldsymbol{A}\hat{\boldsymbol{x}}}-\abs{\boldsymbol{A}\boldsymbol{x}_0}\|_2-\|\boldsymbol{\eta}\|_2 \leq \|\abs{\boldsymbol{A}\hat{\boldsymbol{x}}}-\abs{\boldsymbol{A}\boldsymbol{x}_0}-\boldsymbol{\eta}\|_2 \overset{(a)}{\leq} \|\boldsymbol{\eta}\|_2,
\]
which implies 
\begin{equation}\label{eq:Ahatx}
\|\abs{\boldsymbol{A}\hat{\boldsymbol{x}}}-\abs{\boldsymbol{A}\boldsymbol{x}_0}\|_2 \leq 2\|\boldsymbol{\eta}\|_2.
\end{equation}
Here, the inequality $(a)$ follows from the definition of  $\hat{\boldsymbol{x}}$.
By Theorem \ref{lem: large_distance_2} we obtain that, for any $\delta_1\in (0,\frac{0.05}{\beta_0^{\mathbb H}})$, the following holds
 with probability at least $1-2\exp(-c'\delta_1^2m)$:
\begin{equation}\label{eq:LAlower}
{(\frac{1}{\beta_0^{\mathbb{H}}}-\delta_1)}\cdot   {\text{dist}(\boldsymbol{x},\boldsymbol{y})} \,\,\leq \,\, 
{\frac{1}{\sqrt{m}}\cdot \||\boldsymbol{A}\boldsymbol{x}|-|\boldsymbol{A}\boldsymbol{y}|\|_2}
\end{equation}
for all $\boldsymbol{x}, \boldsymbol{y}\in {\mathbb H}^d$, where $c'$ is a universal constant.
Combining (\ref{eq:Ahatx}) and 
(\ref{eq:LAlower}), we obtain that
\[
\dist(\hat{\boldsymbol{x}}, \boldsymbol{x}_0)\,\,\leq\,\, \frac{2\beta_0^{\mathbb H}}{1-\delta}\cdot \frac{\|\boldsymbol{\eta}\|_2}{\sqrt{m}}.
\]
Here, $\delta:=\delta_1\cdot \beta_0^{\mathbb H}$.
\end{proof}

 \begin{remark}
 In \cite{HX}, it has been established that for the real case, the inequality $\dist(\hat{\boldsymbol{x}}, \boldsymbol{x}_0)\leq C_0\cdot \frac{\|\boldsymbol{\eta}\|_2}{\sqrt{m}} $ holds. Similarly, in the complex case, the corresponding inequality was derived in \cite{Xia}. 
 However, neither \cite{HX} nor \cite{Xia} provide the specific value of $C_0$.
 In our work, by comparing the results presented in \cite{HX,Xia}, we present the precise value of the constant $C_0$ in Corollary \ref{co:budeng}, which is given by $C_0=\frac{2\beta_0^{\mathbb H}}{1-\delta}$. 
  It is noteworthy that as $\delta$ approaches $0$, the constant $C_0$ converges to $2\beta_0^{\mathbb H}$.
  It would be interesting to investigate whether the constant $2\beta_0^{\mathbb H}$ is the tightest possible value.
  \end{remark}

\subsection{Proof of Theorem \ref{lem: large_distance_2}}

To begin with, we present some preliminary results that will aid in proving Theorem \ref{lem: large_distance_2}. The proof of  Lemma \ref{lem: large_distance_1} and Lemma \ref{expectation_result_new} are postponed to Appendix \ref{appendix B}. We believe that the result of Lemma \ref{expectation_result_new} is of independent interest.

\begin{lem}\label{lem: large_distance_1}
Assume that $\boldsymbol{A}=(\boldsymbol{a}_1,\ldots,\boldsymbol{a}_m)^*\in \mathbb{H}^{m\times d}$ is a standard Gaussian random matrix, where $\mathbb{H}=\mathbb{R}$ or $\mathbb{C}$. 
For any $0<\delta<1$, there exist a universal positive constant $c_3$, such that for $m\gtrsim \log(1/\delta)\delta^{-2}d$, with probability at least $1-4\exp(-c_3\delta ^2m)$, the following holds
\begin{equation}\label{xufeb2024-3}
\frac{1}{m}  \sum_{k=1}^{m}|\boldsymbol{y}^*\boldsymbol{a}_k\boldsymbol{a}_k^*\boldsymbol{x}|
\leq \mathbb{E} |\boldsymbol{y}^*\boldsymbol{a}_1\boldsymbol{a}_1^*\boldsymbol{x}|+\delta
\end{equation}
for all unit-norm vectors $\boldsymbol{x},\boldsymbol{y}\in\mathbb{S}_{\mathbb{H}}^{d-1}$.
\end{lem}

\begin{proof}
The proof is presented in Appendix \ref{Proof of lem: large_distance_1}.
\end{proof}

\begin{lem}\label{expectation_result_new}
Assume that $\boldsymbol{a}\in \mathbb{H}^{d}$ is a standard Gaussian random vector, where $\mathbb{H}=\mathbb{R}$ or $\mathbb{C}$. Let $\beta_0^{\mathbb{H}}$ be the constant defined in \eqref{eq:mdlower-new}.
Then for any two unit-norm vectors $\boldsymbol{x},\boldsymbol{y}\in \mathbb{S}_{\mathbb{H}}^{d-1}$, we have
\begin{equation*}
\mathbb{E} |\boldsymbol{y}^*\boldsymbol{a}\boldsymbol{a}^*\boldsymbol{x}| 
\leq 1-\frac{\mathrm{dist}^2(\boldsymbol{x},\boldsymbol{y}) }{2\cdot (\beta_0^{\mathbb{H}})^2}=\begin{cases}
1-(\frac{1}{2}-\frac{1}{\pi})\cdot \mathrm{dist}^2(\boldsymbol{x},\boldsymbol{y}) & \text{if $\mathbb{H}=\mathbb{R}$,}\\
1-(\frac{1}{2}-\frac{\pi}{8})\cdot \mathrm{dist}^2(\boldsymbol{x},\boldsymbol{y}) & \text{if $\mathbb{H}=\mathbb{C}$.}
\end{cases}
\end{equation*}

\end{lem}

\begin{proof}
The proof is presented in Appendix \ref{Proof of expectation_result_new}.
\end{proof}

Now we can present a proof of Theorem \ref{lem: large_distance_2}.

\begin{proof}[Proof of Theorem \ref{lem: large_distance_2}]
Based on Lemma \ref{upperbound-orth}, we have
\begin{equation*}
(L_{\boldsymbol{A}})^2=\min_{\substack{\boldsymbol{x}\in \mathbb{H}^{d},\boldsymbol{y}\in \mathbb{H}^{d} \\ \|\boldsymbol{x}\|_2=1, \|\boldsymbol{y}\|_2\leq 1,\langle \boldsymbol{x},\boldsymbol{y}\rangle=0  }}\frac{\||\boldsymbol{A}\boldsymbol{x}|-|\boldsymbol{A}\boldsymbol{y}|\|_2^2}{\mathrm{dist}^2(\boldsymbol{x},\boldsymbol{y})},
\end{equation*}
so it is enough to estimate the minimum possible ratio between $\||\boldsymbol{A}\boldsymbol{x}|-|\boldsymbol{A}\boldsymbol{y}|\|_2^2$ and $\mathrm{dist}^2(\boldsymbol{x},\boldsymbol{y})$ when $\|\boldsymbol{x}\|_2=1$, $\|\boldsymbol{y}\|_2\leq 1$ and $\langle \boldsymbol{x},\boldsymbol{y}\rangle=0$.

We first consider the case when $\|\boldsymbol{x}\|_2=1$ and $\boldsymbol{y}=\boldsymbol{0}$. In this case,
we use Lemma \ref{complex_fund1} to obtain that for any $\delta_1\in(0,0.01)$, the following 
\begin{equation}\label{eqmayxu4}
\frac{\frac{1}{m}\||\boldsymbol{A}\boldsymbol{x}|-|\boldsymbol{A}\boldsymbol{y}|\|_2^2}{\text{dist}^2(\boldsymbol{x},\boldsymbol{y})}=\frac{1}{m}\frac{\|\boldsymbol{A}\boldsymbol{x}\|_2^2}{\|\boldsymbol{x}\|_2^2}\geq 1-\delta_1,\quad \forall \boldsymbol{x}\in\mathbb{S}_\mathbb{H}^{d-1}
\end{equation}
holds with probability at least $1-2\exp(-c_1\delta_1^2m)$, provided that $m\gtrsim \delta_1^{-2}d$. Here, $c_1$ is the universal positive constant in Lemma \ref{complex_fund1}.

Next, we consider the case when $\|\boldsymbol{x}\|_2=1$, $\|\boldsymbol{y}\|_2\in(0,1]$ and $\langle \boldsymbol{x},\boldsymbol{y}\rangle=0$. We claim that for any $\delta_1\in(0,0.01)$,  if $m\gtrsim  \delta_1^{-2}d$ then with probability at least $1-2\exp(-c_1\delta_1^2m)$, it holds that
\begin{equation}\label{eqmayxu2}
\frac{\frac{1}{m}\||\boldsymbol{A}\boldsymbol{x}|-|\boldsymbol{A}\boldsymbol{y}|\|_2^2}{\text{dist}^2(\boldsymbol{x},\boldsymbol{y})}\geq \frac{\frac{1}{m}\||\boldsymbol{A}{\boldsymbol{x}}|-|\boldsymbol{A}\widetilde{\boldsymbol{y}}|\|_2^2}{\text{dist}^2({\boldsymbol{x}},\widetilde{\boldsymbol{y}})}-2\delta_1	
\end{equation}
for all $\boldsymbol{x},\boldsymbol{y}\in\mathbb{H}^{d}$ with $\|\boldsymbol{x}\|_2=1$, $\|\boldsymbol{y}\|_2\in(0,1]$ and $\langle \boldsymbol{x},\boldsymbol{y}\rangle=0$.
Here, $\widetilde{\boldsymbol{y}}:=\frac{\boldsymbol{y}}{\|\boldsymbol{y}\|_2}\in\mathbb{S}_{\mathbb{H}}^{d-1}$ and $c_1$ is the universal positive constant in Lemma \ref{complex_fund1}. 
Using Lemma \ref{complex_fund1}, we have
\begin{equation*}
\frac{1}{m}\||\boldsymbol{A}\boldsymbol{x}|-|\boldsymbol{A}\widetilde{\boldsymbol{y}}|\|_2^2
=\frac{1}{m}\|\boldsymbol{A}\boldsymbol{x}\|_2^2+ \frac{1}{m}\|\boldsymbol{A}\widetilde{\boldsymbol{y}}\|_2^2-\frac{2}{m}  \sum_{k=1}^{m}|\widetilde{\boldsymbol{y}}^*\boldsymbol{a}_k\boldsymbol{a}_k^*\boldsymbol{x}|\geq 2-2\delta_1-\frac{2}{m}  \sum_{k=1}^{m}|\widetilde{\boldsymbol{y}}^*\boldsymbol{a}_k\boldsymbol{a}_k^*\boldsymbol{x}|.
\end{equation*}
By Lemma \ref{lem: large_distance_1} and Lemma \ref{expectation_result_new}, we further obtain that with probability at least $1-4\exp(-c'\delta_1^2 m)$,
\begin{equation}\label{eqmayxu1}
\begin{aligned}
\frac{1}{m}\||\boldsymbol{A}\boldsymbol{x}|-|\boldsymbol{A}\widetilde{\boldsymbol{y}}|\|_2^2
& \geq 2-2\delta_1-2\cdot \Big( \mathbb{E} |\widetilde{\boldsymbol{y}}^*\boldsymbol{a}_1\boldsymbol{a}_1^*\boldsymbol{x}|+\delta_1\Big)\\
&\geq 2-2\delta_1-2\cdot \Big( 1-\frac{\mathrm{dist}^2(\boldsymbol{x},\widetilde{\boldsymbol{y}}) }{2\cdot (\beta_0^{\mathbb{H}})^2}+\delta_1\Big)
\overset{(a)}=(\frac{1}{ (\beta_0^{\mathbb{H}})^2}-2\delta_1)\cdot \mathrm{dist}^2(\boldsymbol{x},\widetilde{\boldsymbol{y}}),		
\end{aligned}
\end{equation}
provided that $m\gtrsim \log(1/\delta_1)\delta_1^{-2}d$, where $c'$ is a universal  positive constant. 
Here, in equation ($a$) we use the fact that $\mathrm{dist}^2(\boldsymbol{x},\widetilde{\boldsymbol{y}})=\|\boldsymbol{x}-\widetilde{\boldsymbol{y}}\|_2^2=2$, because $\langle \boldsymbol{x},\widetilde{\boldsymbol{y}}\rangle =0$ and $\|\boldsymbol{x}\|_2=\|\widetilde{\boldsymbol{y}}\|_2=1$. 
Combining \eqref{eqmayxu1} and \eqref{eqmayxu2} we have
\begin{equation}\label{eqmayxu3}
\frac{\frac{1}{m}\||\boldsymbol{A}\boldsymbol{x}|-|\boldsymbol{A}\boldsymbol{y}|\|_2^2}{\text{dist}^2(\boldsymbol{x},\boldsymbol{y})}
\geq \frac{\frac{1}{m}\||\boldsymbol{A}{\boldsymbol{x}}|-|\boldsymbol{A}\widetilde{\boldsymbol{y}}|\|_2^2}{\text{dist}^2({\boldsymbol{x}},\widetilde{\boldsymbol{y}})}-2\delta_1 
\geq\frac{1}{ (\beta_0^{\mathbb{H}})^2}-4\delta_1
\end{equation}
for all $\boldsymbol{x},\boldsymbol{y}\in\mathbb{H}^{d}$ with $\|\boldsymbol{x}\|_2=1$, $\|\boldsymbol{y}\|_2\in(0,1]$ and $\langle \boldsymbol{x},\boldsymbol{y}\rangle=0$.
Combining \eqref{eqmayxu4} and \eqref{eqmayxu3}, we obtain that, for any $\delta_1\in(0,0.01)$,
\begin{equation*}
\frac{1}{\sqrt{m}}\cdot L_{\BA}
\geq  \sqrt{\frac{1}{(\beta_0^{\mathbb{H}})^2}-4\delta_1}
\overset{(b)}\geq \frac{1}{\beta_0^{\mathbb{H}}}-5\delta_1
\end{equation*}
holds with probability at least $1-4\exp(-c''\delta_1^2 m)$ for a universal positive constant $c''$. Here, in inequality ($b$) we use $\delta_1<0.01$.
Taking $\delta=5\cdot  \delta_1$, we arrive at our conclusion.

It remains to prove \eqref{eqmayxu2}. Note that if $\langle \boldsymbol{x},\boldsymbol{y}\rangle =0$ then we have $\text{dist}^2(\boldsymbol{x},\boldsymbol{y})=\|\boldsymbol{x}-\boldsymbol{y}\|_2^2=1+\|\boldsymbol{y}\|_2^2$ and $\text{dist}^2({\boldsymbol{x}},\widetilde{\boldsymbol{y}})=\|{\boldsymbol{x}}-\widetilde{\boldsymbol{y}}\|_2^2=2$, where $\widetilde{\boldsymbol{y}}=\frac{\boldsymbol{y}}{\|\boldsymbol{y}\|_2}$. 
Also note that for each $1\leq k\leq m$ we have
\begin{equation}\label{similar-cal}
(|\boldsymbol{a}_k^*\boldsymbol{x}|-|\boldsymbol{a}_k^*\boldsymbol{y}|)^2=|\boldsymbol{a}_k^*(\boldsymbol{x}-\boldsymbol{y})|^2	-2\cdot \big(|\boldsymbol{x}^*\boldsymbol{a}_k\boldsymbol{a}_k^*\boldsymbol{y}|-\mathrm{Re}(\boldsymbol{x}^*\boldsymbol{a}_k\boldsymbol{a}_k^*\boldsymbol{y})\big).
\end{equation}
Therefore, a direct calculation shows that for all $\boldsymbol{x},\boldsymbol{y}\in\mathbb{H}^{d}$ with $\|\boldsymbol{x}\|_2=1$, $\|\boldsymbol{y}\|_2\in(0,1]$ and $\langle \boldsymbol{x},\boldsymbol{y}\rangle=0$,
\begin{equation*}
\begin{aligned}
\frac{\frac{1}{m}\||\boldsymbol{A}\boldsymbol{x}|-|\boldsymbol{A}\boldsymbol{y}|\|_2^2}{\text{dist}^2(\boldsymbol{x},\boldsymbol{y})}
&
=\frac{1}{m}\sum_{k=1}^m\frac{|\boldsymbol{a}_k^*(\boldsymbol{x}-\boldsymbol{y})|^2}{\|\boldsymbol{x}-\boldsymbol{y}\|_2^2}-
\frac{2}{m}\sum_{k=1}^m\frac{|\boldsymbol{x}^*\boldsymbol{a}_k\boldsymbol{a}_k^*\boldsymbol{y}|-\mathrm{Re}(\boldsymbol{x}^*\boldsymbol{a}_k\boldsymbol{a}_k^*\boldsymbol{y})}{1+\|\boldsymbol{y}\|_2^2}\\
&=\frac{1}{m}\sum_{k=1}^m\frac{|\boldsymbol{a}_k^*(\boldsymbol{x}-\boldsymbol{y})|^2}{\|\boldsymbol{x}-\boldsymbol{y}\|_2^2}-
\frac{2}{m}\sum_{k=1}^m\frac{|{\boldsymbol{x}}^*\boldsymbol{a}_k\boldsymbol{a}_k^*\widetilde{\boldsymbol{y}}|-\mathrm{Re}({\boldsymbol{x}}^*\boldsymbol{a}_k\boldsymbol{a}_k^*\widetilde{\boldsymbol{y}})}{\frac{1}{\|\boldsymbol{y}\|_2}+\|\boldsymbol{y}\|_2}\\
&\overset{(a)}\geq \frac{1}{m}\sum_{k=1}^m\frac{|\boldsymbol{a}_k^*(\boldsymbol{x}-\boldsymbol{y})|^2}{\|\boldsymbol{x}-\boldsymbol{y}\|_2^2}-
\frac{2}{m}\sum_{k=1}^m\frac{|{\boldsymbol{x}}^*\boldsymbol{a}_k\boldsymbol{a}_k^*\widetilde{\boldsymbol{y}}|-\mathrm{Re}({\boldsymbol{x}}^*\boldsymbol{a}_k\boldsymbol{a}_k^*\widetilde{\boldsymbol{y}})}{2}\\
&\overset{(b)}=\frac{1}{m}\sum_{k=1}^m\frac{|\boldsymbol{a}_k^*(\boldsymbol{x}-\boldsymbol{y})|^2}{\|\boldsymbol{x}-\boldsymbol{y}\|_2^2}-
\frac{1}{m}\sum_{k=1}^m\frac{|\boldsymbol{a}_k^*({\boldsymbol{x}}-\widetilde{\boldsymbol{y}})|^2}{\|{\boldsymbol{x}}-\widetilde{\boldsymbol{y}}\|_2^2}+
\frac{\frac{1}{m}\||\boldsymbol{A}{\boldsymbol{x}}|-|\boldsymbol{A}\widetilde{\boldsymbol{y}}|\|_2^2}{\text{dist}^2({\boldsymbol{x}},\widetilde{\boldsymbol{y}})}\\
&\overset{(c)}\geq (1-\delta_1)-(1+\delta_1)+\frac{\frac{1}{m}\||\boldsymbol{A}{\boldsymbol{x}}|-|\boldsymbol{A}\widetilde{\boldsymbol{y}}|\|_2^2}{\text{dist}^2({\boldsymbol{x}},\widetilde{\boldsymbol{y}})}
=\frac{\frac{1}{m}\||\boldsymbol{A}{\boldsymbol{x}}|-|\boldsymbol{A}\widetilde{\boldsymbol{y}}|\|_2^2}{\text{dist}^2({\boldsymbol{x}},\widetilde{\boldsymbol{y}})}-2\delta_1.
\end{aligned}	
\end{equation*}
Here, inequality ($a$) follows from Cauchy-Schwarz inequality. In equation ($b$) we use the fact that $\text{dist}^2({\boldsymbol{x}},\widetilde{\boldsymbol{y}})=\|{\boldsymbol{x}}-\widetilde{\boldsymbol{y}}\|_2^2=2$ and use a similar calculation in \eqref{similar-cal}. Inequality ($c$) follows from Lemma \ref{complex_fund1}, which holds with probability at least $1-2\exp(-c_1\delta_1^2m)$ when $m\gtrsim \delta_1^{-2}d$. Therefore, we arrive at \eqref{eqmayxu2}. This completes the proof.

\end{proof}

\section*{Acknowledgement}

The authors would like to thank Radu Balan for bringing \cite{AAFG24} to our attention and for providing many useful suggestions on an initial version of this work, resulting in significant improvements to the manuscript.

\Addresses

\newpage
\appendix
\counterwithin{equation}{section}

\section{Proof of Lemma \ref{prop-calculation}}\label{Proof of prop-calculation}

In this section we present the proof of  Lemma \ref{prop-calculation}. We first prove the following lemma.

\begin{lem}\label{calculate1}
Let $l\geq 0$ be an integer. Then, for any $\beta,\theta\in \R$,  we have
\begin{equation}\label{eq:Tm}
T_l(\beta,\theta):=\sum_{j=0}^{l}\sin( 2j\beta+2\theta) =\sin2\theta\Big(\frac{\sin((2l+1)\beta)}{2\sin(\beta)}+\frac{1}{2}\Big)
+\frac{\cos 2\theta\sin((l+1)\beta)\sin(l\beta)}{\sin(\beta)}.
\end{equation}

\end{lem}
\begin{proof}
By Lagrange's trigonometric identities we have
\begin{equation}\label{eq1:4}
\sum_{j=0}^{l}\cos( j\cdot x)=\frac{\sin(\frac{2l+1}{2}x)}{2\sin(\frac{x}{2})}+\frac{1}{2}\quad\text{and}\quad 
\sum_{j=0}^{l}\sin( j\cdot x)=\frac{\sin(\frac{l+1}{2}x)\cdot \sin(\frac{l}{2}x)}{\sin(\frac{x}{2})}.	
\end{equation}	
Note that 
\begin{equation*}
T_l(\beta,\theta)=\sum_{j=0}^{l}
\left(\sin 2j \beta\cos2\theta+\cos2j \beta\sin2\theta\right)=
\cos2\theta\sum_{j=0}^{l}\sin 2j \beta
+\sin2\theta\sum_{j=0}^{l}\cos2j \beta.
\end{equation*}
Substituting \eqref{eq1:4} with $x=2\beta$ into the above equation we arrive at our conclusion.
\end{proof}

Now we can present a proof of Lemma \ref{prop-calculation}.

\begin{proof}[Proof of Lemma \ref{prop-calculation}]
We divide the proof into two cases.

\textbf{Case 1: $m=2k$ is even.}
Note that $G_m(\theta)=G_m(\theta+\frac{\pi}{m})$ for all $\theta\in\mathbb{R}$. 
To prove the lemma, it suffices to find the maximum value of $G_m(\theta)$ when $\theta\in[0,\frac{\pi}{m}]$. Note that for each $\theta\in[0,\frac{\pi}{m}]$, we have
\begin{equation}\label{eq1:5}
\sin\Big(\frac{2j\pi}{m}+2\theta\Big)\begin{cases}
\geq 0 & \text{if $j\in \{0,\ldots,k-1\}$,}\\
\leq 0	& \text{if $j\in\{k,k+1,\ldots,2k-1\}$.}
\end{cases}	
\end{equation}
Therefore, we use Lemma \ref{calculate1} to obtain that
\begin{equation}\label{eqGeq:2}
G_m(\theta)=2T_{k-1}\Big(\frac{\pi}{m},\theta\Big)-T_{2k-1}\Big(\frac{\pi}{m},\theta\Big)=\frac{2\cos(2\theta-\frac{\pi}{m})}{\sin\frac{\pi}{m}},
\end{equation}
where $T_l(\beta,\theta)$ is defined as in \eqref{eq:Tm}.
Then we have
\begin{equation*}
\max_{0\leq \theta\leq \frac{\pi}{m}} G_m(\theta)=G_m\Big(\frac{\pi}{2m}\Big)=\frac{2}{\sin\frac{\pi}{m}}.	
\end{equation*}
Using the fact that $G_m(\theta)$ has a period of $\frac{\pi}{m}$, we arrive at the desired conclusion.

\textbf{Case 2: $m=2k+1$ is odd.}
A simple calculation shows that for each $j\in\mathbb{Z}$ and for each $\theta\in\mathbb{R}$, we have
$$\Big|\sin\Big(\frac{2j\pi}{m}+2(\theta+\frac{\pi}{2m})\Big)\Big|=\Big|\sin\Big(\frac{2(j-k)\pi}{m}+2\theta\Big)\Big|,$$
so one can easily check that $G_m(\theta)=G_m(\theta+\frac{\pi}{2m})$ for all $\theta\in\mathbb{R}$. 
Therefore, to prove the lemma, it suffices to find the maximum value of $G_m(\theta)$ when $\theta\in[0,\frac{\pi}{2m}]$.
For any $0\leq\theta\leq \frac{\pi}{2m}$, we have
\begin{equation}\label{eq2:5}
\sin\Big(\frac{2j\pi}{m}+2\theta\Big)\begin{cases}
\geq 0 & \text{if $j\in \{0,\ldots,k\}$,}\\
\leq 0	& \text{if $j\in\{k+1,k+2,\ldots,2k\}$,}
\end{cases}	
\end{equation}
Then,
\begin{equation}\label{eqGeq:5}
G_m(\theta)=2T_{k}\Big(\frac{\pi}{m},\theta\Big)-T_{2k}\Big(\frac{\pi}{m},\theta\Big)
=\zfrac{1}{\sin\frac{\pi}{2m}}\cdot \cos(2\theta- \frac{\pi}{2m}),
\end{equation}
where we use Lemma \ref{calculate1} in the last equation.
Therefore, we have
\begin{equation*}
\max_{0\leq \theta\leq \frac{\pi}{2m}} G_m(\theta)=G_m\Big(\frac{\pi}{4m}\Big)=\frac{1}{\sin\frac{\pi}{2m}}.
\end{equation*}
Using the fact that $G_m(\theta)$ has a period of $\frac{\pi}{2m}$, we arrive at the desired conclusion.

\end{proof}

\section{Proof of Lemma \ref{lem: large_distance_1} and Lemma \ref{expectation_result_new}}\label{appendix B}

\subsection{Proof of Lemma \ref{lem: large_distance_1}}\label{Proof of lem: large_distance_1}

We first prove Lemma \ref{lem: large_distance_1}. We need the following lemma.

\begin{lem}\cite{vershynin}\label{Bernstein} Let $\xi_1,\ldots,\xi_m$ be i.i.d. sub-exponential
random variables and $K:=\max_j \|\xi_j\|_{\psi_1}$, where
$\|\cdot\|_{\psi_1}:=\sup_{p\geq 1} p^{-1}(\mathbb{E}|\cdot|^p)^{1/p}$. Then for every $\epsilon>0$, we
have
\begin{equation*}
\mathbb{P}\left(\left|\frac{1}{m}\sum_{j=1}^m\xi_j-\frac{1}{m}\mathbb{E}\left(\sum_{j=1}^m\xi_j\right)\right|\geq \epsilon\right)\leq 2\exp\left(-c_0 m\min\left(\frac{\epsilon^2}{K^2},\frac{\epsilon}{K}\right)\right),
\end{equation*}
where $c_0>0$ is an absolute constant.
\end{lem}

Now we can present a proof of Lemma \ref{lem: large_distance_1}.  

\begin{proof}[Proof of Lemma \ref{lem: large_distance_1}]
We first prove the result for any two fixed $\boldsymbol{x},\boldsymbol{y}\in\mathbb{S}_{\mathbb{H}}^{d-1}$, and then apply an $\epsilon$-net argument to derive a uniform bound in \eqref{xufeb2024-3}.
Now let $\boldsymbol{x},\boldsymbol{y}\in\mathbb{S}_{\mathbb{H}}^{d-1}$ be two fixed vectors.
Note that the terms $|\boldsymbol{y}^*\boldsymbol{a}_k\boldsymbol{a}_k^*\boldsymbol{x}|$ are independent sub-exponential random variables with the maximal sub-exponential norm $C_{\psi}$ for some positive absolute constant $C_{\psi}$. Applying Bernstein's inequality in Lemma \ref{Bernstein} we obtain that for any $0<\delta\leq 1$, 
\begin{equation}\label{xufeb2024-2}
\frac{1}{m}  \sum_{k=1}^{m}|\boldsymbol{y}^*\boldsymbol{a}_k\boldsymbol{a}_k^*\boldsymbol{x}|	\leq \mathbb{E}|\boldsymbol{y}^*\boldsymbol{a}_1\boldsymbol{a}_1^*\boldsymbol{x}|+\frac{\delta}{2}
\end{equation}
holds with probability at least $1-2\exp(-c'\delta^2m)$. Here, $c'$ is a universal positive constant.

Next, we give a uniform bound for \eqref{xufeb2024-3}. 
Let ${\mathcal{N}}_{\epsilon}$ be an $\epsilon$-net of the unit sphere $\mathbb{S}_{\mathbb{H}}^{d-1}=\{\boldsymbol{z}\in \mathbb{H}^{d}\ :\ \|\boldsymbol{z}\|_2=1\}$ with cardinality $\#{\mathcal{N}}_{\epsilon}\leq (3/\epsilon)^{2d}$. Here, $\epsilon>0$ is a small constant to be specified later.
Note that the event
\begin{equation*}
\mathcal{E}:=\bigg\{\frac{1}{m}  \sum_{k=1}^{m}|\boldsymbol{y}^*\boldsymbol{a}_k\boldsymbol{a}_k^*\boldsymbol{x}|
\leq 
\mathbb{E}|\boldsymbol{y}^*\boldsymbol{a}_1\boldsymbol{a}_1^*\boldsymbol{x}|+\frac{\delta}{2},
\quad \forall \boldsymbol{x},\boldsymbol{y}\in \mathcal{N}_{\epsilon}\bigg\}
\end{equation*}
holds with probability at least $1-2\exp(-c'\delta^2m)\cdot (3/\epsilon)^{4d}$.
For any $\boldsymbol{x},\boldsymbol{y}\in \mathbb{S}_{\mathbb{H}}^{d-1}$, there exist $\boldsymbol{x}_1,\boldsymbol{y}_1\in \mathcal{N}_{\epsilon}$ such that $\|\boldsymbol{x}-\boldsymbol{x}_1\|_2\leq\epsilon$ and $\|\boldsymbol{y}-\boldsymbol{y}_1\|_2\leq\epsilon$. Conditioned on the event $\mathcal{E}$, we have
\begin{equation}\label{xufeb2024-6}
\begin{aligned}
\frac{1}{m}  \sum_{k=1}^{m}|\boldsymbol{y}^*\boldsymbol{a}_k\boldsymbol{a}_k^*\boldsymbol{x}|
&\leq \frac{1}{m}  \sum_{k=1}^{m}|\boldsymbol{y}_1^*\boldsymbol{a}_k\boldsymbol{a}_k^*\boldsymbol{x}_1|+\frac{1}{m}  \sum_{k=1}^{m}\Big| |\boldsymbol{y}^*\boldsymbol{a}_k\boldsymbol{a}_k^*\boldsymbol{x}|-|\boldsymbol{y}_1^*\boldsymbol{a}_k\boldsymbol{a}_k^*\boldsymbol{x}_1|\Big|\\
&\leq \mathbb{E}|\boldsymbol{y}_1^*\boldsymbol{a}_1\boldsymbol{a}_1^*\boldsymbol{x}_1|+\frac{\delta}{2}+\frac{1}{m}  \sum_{k=1}^{m}\Big| |\boldsymbol{y}^*\boldsymbol{a}_k\boldsymbol{a}_k^*\boldsymbol{x}|-|\boldsymbol{y}_1^*\boldsymbol{a}_k\boldsymbol{a}_k^*\boldsymbol{x}_1|\Big|.
\end{aligned}
\end{equation}
A simple calculation shows that for each $1\leq k\leq m$,
\begin{equation}\label{xumarch-88}
\begin{aligned}
\Big| |\boldsymbol{y}^*\boldsymbol{a}_k\boldsymbol{a}_k^*\boldsymbol{x}|-|\boldsymbol{y}_1^*\boldsymbol{a}_k\boldsymbol{a}_k^*\boldsymbol{x}_1|\Big|
&\leq 	\big| |\boldsymbol{y}^*\boldsymbol{a}_k\boldsymbol{a}_k^*\boldsymbol{x}|-|\boldsymbol{y}_1^*\boldsymbol{a}_k\boldsymbol{a}_k^*\boldsymbol{x}|\big|
+\big| |\boldsymbol{y}_1^*\boldsymbol{a}_k\boldsymbol{a}_k^*\boldsymbol{x}|-|\boldsymbol{y}_1^*\boldsymbol{a}_k\boldsymbol{a}_k^*\boldsymbol{x}_1|\big|\\
&=|\boldsymbol{a}_k^*\boldsymbol{x}|\cdot \big| |\boldsymbol{y}^*\boldsymbol{a}_k|-|\boldsymbol{y}_1^*\boldsymbol{a}_k|\big|
+|\boldsymbol{a}_k^*\boldsymbol{y}_1|\cdot \big| |\boldsymbol{x}^*\boldsymbol{a}_k|-|\boldsymbol{x}_1^*\boldsymbol{a}_k|\big|\\
&\leq |\boldsymbol{a}_k^*\boldsymbol{x}|\cdot | (\boldsymbol{y}-\boldsymbol{y}_1)^*\boldsymbol{a}_k|
+|\boldsymbol{a}_k^*\boldsymbol{y}_1|\cdot |(\boldsymbol{x}-\boldsymbol{x}_1)^*\boldsymbol{a}_k|,
\end{aligned}	
\end{equation}
Then, with probability at least $1-2\exp(-c_1 m)$, we have
\begin{equation}\label{xufeb2024-4}
\begin{aligned}
\frac{1}{m}  \sum_{k=1}^{m}\Big| |\boldsymbol{y}^*\boldsymbol{a}_k\boldsymbol{a}_k^*\boldsymbol{x}|-|\boldsymbol{y}_1^*\boldsymbol{a}_k\boldsymbol{a}_k^*\boldsymbol{x}_1|\Big|&\leq \frac{1}{m}  \sum_{k=1}^{m} |\boldsymbol{a}_k^*\boldsymbol{x}|\cdot | (\boldsymbol{y}-\boldsymbol{y}_1)^*\boldsymbol{a}_k|
+|\boldsymbol{a}_k^*\boldsymbol{y}_1|\cdot |(\boldsymbol{x}-\boldsymbol{x}_1)^*\boldsymbol{a}_k| \\
&\overset{(a)}\leq \frac{1}{m}\| \boldsymbol{A}\boldsymbol{x}\|_2\cdot \| \boldsymbol{A}(\boldsymbol{y}-\boldsymbol{y}_1)\|_2 	
+\frac{1}{m}\| \boldsymbol{A}\boldsymbol{y}_1\|_2\cdot \| \boldsymbol{A}(\boldsymbol{x}-\boldsymbol{x}_1)\|_2 \overset{(b)}\leq 4\epsilon,	
\end{aligned}
\end{equation}
where $m\gtrsim d$ and $c_1$ is the universal positive constant in Lemma \ref{complex_fund1}. Here, we use Cauchy-Schwarz inequality in inequality ($a$). Inequality ($b$) follows from the fact that $\frac{1}{\sqrt{m}}\|\boldsymbol{A}\|_2\leq \sqrt{2}$ with probability at least $1-2\exp(-c_1m)$, implied by Lemma \ref{complex_fund1}. 
Moreover, by \eqref{xumarch-88} we have
\begin{equation}\label{xufeb2024-5}
\begin{aligned}
\mathbb{E}|\boldsymbol{y}_1^*\boldsymbol{a}_1\boldsymbol{a}_1^*\boldsymbol{x}_1|
&\leq \mathbb{E}|\boldsymbol{y}^*\boldsymbol{a}_1\boldsymbol{a}_1^*\boldsymbol{x}| + 
\mathbb{E} |\boldsymbol{a}_1^*\boldsymbol{x}|\cdot | (\boldsymbol{y}-\boldsymbol{y}_1)^*\boldsymbol{a}_1|
+\mathbb{E} |\boldsymbol{a}_1^*\boldsymbol{y}_1|\cdot |(\boldsymbol{x}-\boldsymbol{x}_1)^*\boldsymbol{a}_1|\\
&\leq \mathbb{E}|\boldsymbol{y}^*\boldsymbol{a}_1\boldsymbol{a}_1^*\boldsymbol{x}|+
 \sqrt{\mathbb{E} |\boldsymbol{a}_1^*\boldsymbol{x}|^2 \cdot \mathbb{E} | (\boldsymbol{y}-\boldsymbol{y}_1)^*\boldsymbol{a}_1|^2}
 + \sqrt{\mathbb{E} |\boldsymbol{a}_1^*\boldsymbol{y}|^2 \cdot \mathbb{E} | (\boldsymbol{x}-\boldsymbol{x}_1)^*\boldsymbol{a}_1|^2}\\
 &= \mathbb{E}|\boldsymbol{y}^*\boldsymbol{a}_1\boldsymbol{a}_1^*\boldsymbol{x}|
 + \| \boldsymbol{y}-\boldsymbol{y}_1\|_2 + \| \boldsymbol{x}-\boldsymbol{x}_1\|_2
 \leq \mathbb{E}|\boldsymbol{y}^*\boldsymbol{a}_1\boldsymbol{a}_1^*\boldsymbol{x}|+2\epsilon. 
\end{aligned}	
\end{equation} 
Take $\epsilon=\frac{\delta}{12}$. Plugging \eqref{xufeb2024-4} and \eqref{xufeb2024-5} into \eqref{xufeb2024-6}, we obtain that the inequality
\begin{equation}
\frac{1}{m}  \sum_{k=1}^{m}|\boldsymbol{y}^*\boldsymbol{a}_k\boldsymbol{a}_k^*\boldsymbol{x}|
\leq  \mathbb{E}|\boldsymbol{y}^*\boldsymbol{a}_1\boldsymbol{a}_1^*\boldsymbol{x}|+2\epsilon+\frac{\delta}{2}+4\epsilon 
= 
\mathbb{E}|\boldsymbol{y}^*\boldsymbol{a}_1\boldsymbol{a}_1^*\boldsymbol{x}|+\delta
\end{equation}
holds with probability at least
\begin{equation*}
1-2\exp(-c_1 m)-2\exp(-c'\delta^2m)\cdot (3/\epsilon)^{4d}\geq 1-4\exp(-c'' \delta^2m),
\end{equation*}
provided that $m\gtrsim \log(1/\delta)\delta^{-2}d$. Here, $c''$ is a universal positive constant. Take $c_3=c''$. This completes the proof.

\end{proof}

\subsection{Proof of Lemma \ref{expectation_result_new}}\label{Proof of expectation_result_new}

We next prove Lemma \ref{expectation_result_new}.

\begin{proof}[Proof of Lemma \ref{expectation_result_new}]
Note that we have $|(\exp(\mathrm{i}\theta)\boldsymbol{y})^*\boldsymbol{a}\boldsymbol{a}^*\boldsymbol{x}|=|\boldsymbol{y}^*\boldsymbol{a}\boldsymbol{a}^*\boldsymbol{x}|$ and $\mathrm{dist}(\boldsymbol{x},\boldsymbol{y})=\mathrm{dist}(\boldsymbol{x},\exp(\mathrm{i}\theta)\boldsymbol{y})$  for any $\theta\in \mathbb{R}$. Therefore, without loss of generality we can assume that  $\langle \boldsymbol{x},\boldsymbol{y}\rangle\geq0$. 
Leveraging the rotational invariance characteristic of the Gaussian distribution, we can choose
$\boldsymbol{x}=[1,0,\ldots,0]^{\top}$ and $\boldsymbol{y}=[\cos\theta,\sin\theta,0,\ldots,0]^{\top}$, where $\cos\theta=\langle \boldsymbol{x},\boldsymbol{y}\rangle$ and $\theta\in[0,\pi/2]$. 
Then we have $\mathrm{dist}^2(\boldsymbol{x},\boldsymbol{y})=\|\boldsymbol{x}-\boldsymbol{y}\|_2^2=2-2\cos\theta$. 
In the following, we will estimate $\mathbb{E} |\boldsymbol{y}^*\boldsymbol{a}\boldsymbol{a}^*\boldsymbol{x}|$, dividing the proof into two cases.

\textbf{Case 1: $\mathbb{H}=\mathbb{R}$.}
Denote the first two entries of $\boldsymbol{a}$ as $a$ and $b$. Then we have 
$a,b\sim\mathcal{N}(0,1)$. A simple calculation shows that  
\begin{equation*}
\begin{aligned}
\mathbb{E} |\boldsymbol{y}^*\boldsymbol{a}\boldsymbol{a}^*\boldsymbol{x}|=\mathbb{E}|a(a\cos\theta+b\sin\theta)|
&=\frac{1}{2\pi}\int_{-\infty}^{+\infty}\int_{-\infty}^{+\infty}|u(u\cos \theta +v\sin \theta)|\exp\left(-\frac{u^2+v^2}{2}\right)\mathrm{d}u\mathrm{d}v\\
&=\frac{1}{2\pi}\int_{0}^{+\infty}r^3\exp(-r^2/2)dr\int_{0}^{2\pi}\left|\cos \theta \sin^2 \phi+\sin \theta \sin\phi \cos\phi\right|\mathrm{d}\phi\\
&=\frac{2}{\pi}(\sin\theta+(\frac{\pi}{2}-\theta)\cos\theta).
\end{aligned}
\end{equation*}
Here, the second equality above is based on the polar transformation, that is, $u=r\cos\phi$, $v=r\sin \phi$ with $r\in [0,+\infty)$, $\phi\in [0,2\pi)$ and $\mathrm{d}u\mathrm{d}v=r\cdot \mathrm{d}r\mathrm{d}\phi$.
To establish the lemma in the real case, it suffices to prove that
\begin{equation*}
g(\theta):=1-\frac{2-2\cos\theta }{2\cdot (\beta_0^{\mathbb{R}})^2}-  \frac{2}{\pi}(\sin\theta+(\frac{\pi}{2}-\theta)\cos\theta)\geq0, \quad \forall \theta\in[0, \frac{\pi}{2}].
\end{equation*}
A direct calculation shows that $g(0)=g(\frac{\pi}{2})=0$, $g'(\theta)=\frac{2}{\pi}(1-\theta)\sin\theta\geq 0$ when $\theta\in[0,1]$, and $g'(\theta)\leq 0$ when $\theta\in[1,\frac{\pi}{2}]$, which immediately implies that $g(\theta)\geq0$ for all $\theta\in[0, \frac{\pi}{2}]$.

\textbf{Case 2: $\mathbb{H}=\mathbb{C}$.} 
Denote the first two entries of $\boldsymbol{a}$ as $\frac{a}{\sqrt{2}}$ and $\frac{b}{\sqrt{2}}$, where 
$a=a_{1}+a_{2}\mathrm{i}$ and $b=b_{1}+b_{2}\mathrm{i}$. Here,
$a_{1},a_{2},b_{1},b_{2}\sim\mathcal{N}(0,1)$. Then, we have
\begin{equation*}
\small
\begin{split}
h(\theta)&:=\mathbb{E} |\boldsymbol{y}^*\boldsymbol{a}\boldsymbol{a}^*\boldsymbol{x}|=\frac{1}{2} \mathbb{E}|a(a\cos\theta+b\sin\theta)|\\
&= \frac{1}{2} \mathbb{E}\left(\sqrt{a_{1}^{2}+a_{2}^{2}}\cdot\sqrt{(a_{1}^{2}+a_{2}^{2})\cos^{2}\theta+(b_{1}^{2}+b_{2}^{2})\sin^{2}\theta+\sin2\theta(a_{1}b_{1}+a_{2}b_{2})}\right)\\
&=\frac{1}{2} \int_{-\infty}^{\infty}\int_{-\infty}^{\infty}\int_{-\infty}^{\infty}\int_{-\infty}^{\infty}\frac{1}{(\sqrt{2\pi})^4}f_1(u_1,u_2,v_1,v_2,\theta) \exp\left(-\frac{u_1^2+u_2^2+v_1^2+v_2^2}{2}\right)\mathrm{d}u_1\mathrm{d}u_2\mathrm{d}v_1\mathrm{d}v_2,
\end{split}
\end{equation*}
where \begin{equation*}f_1(u_1,u_2,v_1,v_2,\theta)=\sqrt{u_{1}^{2}+u_{2}^{2}}\cdot\sqrt{(u_{1}^{2}+u_{2}^{2})\cos^{2}\theta+(v_{1}^{2}+v_{2}^{2})\sin^{2}\theta+(u_{1}v_{1}+u_{2}v_{2})\sin2\theta}.\end{equation*}
To establish the lemma in the complex case, it suffices to prove that
\begin{equation*}
h(\theta)\leq 1-\frac{2-2\cos\theta }{2\cdot (\beta_0^{\mathbb{C}})^2}= \cos\theta+\frac{\pi}{4}\cdot (1-\cos\theta), \quad \forall  \theta\in[0, \frac{\pi}{2}].
\end{equation*}
We claim that for all $\theta\in[0, \frac{\pi}{2}]$,
\begin{equation}\label{eq1:2024}
\begin{aligned}
h(\theta)&\overset{(a)}=	\frac{1}{4\pi} \iint_{\mathbb{S}^{2}} \sqrt{1+x\cdot \cos\theta-y\cdot \sin\theta}\cdot \sqrt{1+x\cdot \cos\theta+y\cdot \sin\theta}\ 	  \mathrm{d}S\\
&\overset{(b)} \leq \frac{1}{4\pi} \iint_{\mathbb{S}^{2}} \Big( \cos\theta\cdot  (1+x)+(1-\cos\theta)\cdot  \sqrt{1-y^2}\Big)\ 	  \mathrm{d}S.
\end{aligned}
\end{equation}
Here, $\mathbb{S}^{2}=\{(x,y,z)^T\in\mathbb{R}^3:x^2+y^2+z^2=1\}$ denotes the unit sphere in $\mathbb{R}^3$.
Noting that 
\begin{equation*}
\iint_{\mathbb{S}^{2}}  (1+x)\ 	  \text{d}S=4\pi \quad\text{and}\quad 	\iint_{\mathbb{S}^{2}}\sqrt{1-y^2}\ 	  \text{d}S=\pi^2,
\end{equation*}
we arrive at the desired conclusion:
\begin{equation*}
h(\theta)\leq \frac{1}{4\pi}\big(\cos\theta\cdot 4\pi +(1-\cos\theta)\cdot \pi^2 \big)=\cos\theta+\frac{\pi}{4}\cdot (1-\cos\theta).
\end{equation*}

It remains to prove \eqref{eq1:2024}. We first prove equation  $(a)$ in \eqref{eq1:2024}. 
Taking $u_{1}=\rho\cos\varphi\cos\varphi_{1}$, $u_{2}=\rho\cos\varphi\sin\varphi_{1}$,
$v_{1}=\rho\sin\varphi\cos\varphi_{2}$, $v_{2}=\rho\sin\varphi\sin\varphi_{2}$
with $\rho\in [0,+\infty)$, $\varphi_1,\varphi_2\in[0,2\pi)$, $\varphi\in[0,\pi/2)$ and  $\mathrm{d}u_1\mathrm{d}u_2\mathrm{d}v_1\mathrm{d}v_2=\rho^3\cos\varphi\sin\varphi \cdot \mathrm{d}\rho\mathrm{d}\varphi\mathrm{d}\varphi_1\mathrm{d}\varphi_2$, 
we can use the polar transformation to simplify  $h(\theta)$ as
\begin{equation*}
\begin{aligned}
h(\theta)=& \frac{1}{8\pi^{2}}\int_{0}^{+\infty}\int_{0}^{\pi/2}\int_{0}^{2\pi}\int_{0}^{2\pi}\rho^{5}e^{-\frac{\rho^{2}}{2}}\cos^{2}\varphi\sin\varphi \cdot f_2(\varphi,\varphi_1,\varphi_2,\theta)\mathrm{d}\varphi_{1}\mathrm{d}\varphi_{2}\mathrm{d}\varphi \mathrm{d}\rho\\
= & \frac{1}{\pi^{2}}\int_{0}^{\pi/2}\int_{0}^{2\pi}\int_{0}^{2\pi}\cos^{2}\varphi\sin\varphi\cdot f_2(\varphi,\varphi_1,\varphi_2,\theta)\mathrm{d}\varphi_{1}\mathrm{d}\varphi_{2}\mathrm{d}\varphi,
\end{aligned}
\end{equation*}
where
\begin{equation*}
\begin{aligned}
f_2(\varphi,\varphi_1,\varphi_2,\theta)&=\sqrt{\cos^{2}\theta\cos^{2}\varphi+\sin^{2}\theta\sin^{2}\varphi+\sin2\theta\cos\varphi\sin\varphi\cos(\varphi_{1}-\varphi_{2})}\\
&=\frac{\sqrt{2}}{2}\cdot \sqrt{1+\cos2\theta\cos2\varphi+\sin2\theta\sin2\varphi\cos(\varphi_{1}-\varphi_{2})}.
\end{aligned}	
\end{equation*}
Using the fact that $f_2(\varphi,\varphi_1,\varphi_2,\theta)$ has period $2\pi$ in $\varphi_2$, we can  obtain
\begin{equation*}
h(\theta)=\frac{\sqrt{2}}{\pi}\int_{0}^{\pi/2}\int_{0}^{2\pi}\cos^{2}\varphi\sin\varphi\sqrt{1+\cos2\theta\cos2\varphi+\sin2\theta\sin2\varphi\cos\varphi_{1}}\mathrm{d}\varphi_{1}\mathrm{d}\varphi.
\end{equation*}
By letting $x=\cos2\varphi$, $y=\sin2\varphi\cos\varphi_1$, $z=\sin2\varphi\sin\varphi_1$,
we further have
\begin{equation}\label{eq2:2024}
h(\theta)=\frac{1}{4\pi}  \iint_{\mathbb{S}^{2}} \sqrt{1+x}\cdot \sqrt{1+x\cdot \cos2\theta+y\cdot \sin2\theta}\ 	  \text{d}S=\frac{1}{4\pi}  \iint_{\mathbb{S}^{2}} \sqrt{1+\langle \boldsymbol{x},\boldsymbol{v}_1\rangle}\cdot \sqrt{1+\langle \boldsymbol{x},\boldsymbol{v}_2\rangle}\ 	  \text{d}S,
\end{equation}
where $\boldsymbol{x}=(x,y,z)^T$, $\boldsymbol{v}_1=(1,0,0)^T$ and $\boldsymbol{v}_2=(\cos2\theta,\sin2\theta,0)^T$. Noting that for any rotation matrix $\boldsymbol{P}\in\mathbb{R}^{3\times 3}$, we have
\begin{equation}\label{eq3:2024}
\iint_{\mathbb{S}^{2}} \sqrt{1+\langle \boldsymbol{x},\boldsymbol{v}_1\rangle}\cdot \sqrt{1+\langle \boldsymbol{x},\boldsymbol{v}_2\rangle}\ 	  \text{d}S=\iint_{\mathbb{S}^{2}} \sqrt{1+\langle \boldsymbol{P}\boldsymbol{x},\boldsymbol{v}_1\rangle}\cdot \sqrt{1+\langle \boldsymbol{P}\boldsymbol{x},\boldsymbol{v}_2\rangle}\ 	  \text{d}S.
\end{equation}
Taking
\begin{equation*}
\boldsymbol{P}=\begin{bmatrix}
\cos\theta & -\sin\theta & 0\\
\sin\theta & \cos\theta & 0\\
0&0&1
\end{bmatrix}
\end{equation*}	
and substituting \eqref{eq3:2024} into \eqref{eq2:2024}, we arrive at equation $(a)$ in \eqref{eq1:2024}. 

We next prove inequality $(b)$ in \eqref{eq1:2024}. It is enough to show that for all $\theta\in[0,\frac{\pi}{2}]$ and for all $x,y\in\mathbb{R}$ with $x^2+y^2\leq 1$, we have
\begin{equation*}
A:=\Big(\cos\theta\cdot  (1+x)+(1-\cos\theta)\cdot  \sqrt{1-y^2}\Big)^2-	\Big(1+x\cdot \cos\theta-y\cdot \sin\theta\Big)\Big(1+x\cdot \cos\theta+y\cdot \sin\theta\Big)\geq 0.
\end{equation*}
A direct calculation shows
\begin{equation*}
\begin{aligned}
A=2\cos\theta(1-\cos\theta)	\cdot \Big((1+x)\sqrt{1-y^2}-(1+x-y^2) \Big).
\end{aligned}	
\end{equation*}
Noting that $x^2+y^2\leq1$, we have
$1+x\geq 0$  and 
\begin{equation*}
(1+x)^2(1-y^2)-(1+x-y^2)^2=y^2(1-x^2-y^2)\geq 0,	
\end{equation*}
which implies that $$(1+x)\sqrt{1-y^2}\geq |1+x-y^2|.$$ Combining with $\cos\theta(1-\cos\theta)\geq0$ for each $\theta\in[0,\frac{\pi}{2}]$, we immediately obtain that $A\geq 0$. This completes the proof of \eqref{eq1:2024}.

\end{proof}

\end{document}